\documentclass[10pt, conference, compsocconf]{IEEEtran}
\usepackage{float}
\usepackage{amsmath}
\usepackage{amssymb}
\usepackage{amsthm}
\usepackage{tabularx}
\usepackage{graphicx}
\usepackage{subfigure}
\usepackage{authblk}
\usepackage[linesnumbered, ruled]{algorithm2e}

\usepackage{algorithmic}
\usepackage{tabularx}
\usepackage{booktabs}
\usepackage{amsfonts}
\usepackage{dsfont}
\usepackage[all]{xy}

\def\squarebox#1{\hbox to #1{\hfill\vbox to #1{\vfill}}}

\newtheorem{thm}{Theorem}

\begin{document}

\title{Flow Demands Oriented Node Placement in Multi-Hop Wireless Networks}


\author{Zimu Yuan  \\
Institute of Computing Technology, CAS, China\\
\{zimu.yuan\}@gmail.com
}

\maketitle \thispagestyle{empty}

\begin{abstract}
In multi-hop wireless networks, flow demands mean that some nodes have routing demands of transmitting their data to other nodes with a certain level of transmission rate. When a set of nodes have been deployed with flow demands, it is worth to know how to construct paths to satisfy these flow demands with nodes placed as few as possible. In this paper, we study this flow demands oriented node placement problem that has not been addressed before. In particular, we divide and conquer the problem by three steps: calculating the maximal flow for single routing demand, calculating the maximal flow for multiple routing demands, and finding the minimal number of nodes for multiple routing demands with flow requirement. During the above solving procedure, we prove that the second and third step are NP-hard and propose two algorithms that have polynomial-time complexity. The proposed algorithms are evaluated under practical scenarios. The experiments show that the proposed algorithms can achieve satisfactory results on both flow demands and total number of wireless nodes.
\end{abstract}



\section{Introduction}
Multi-hop wireless networks have gained a lot of attentions in the past few years. One of key design issues in multi-hop wireless networks is node placement. By careful node placement, we can make multi-hop wireless networks achieve special design goals. For example, studies on node placement are related to wide topics such as network traffic \cite{SAC13_throughput}\cite{MNA08_throughput}\cite{JSAC12_throughput}, network coverage \cite{WiOpt03_Coverage}\cite{Mobi04_Coverage}\cite{Info03_Coverage}, network survivability \cite{Info06_survivability}\cite{TOC07_survivability}\cite{TON10_survivability}, fault-tolerant \cite{HPSR04_fault}\cite{TMC07_fault}\cite{Info07_fault}, energy saving \cite{ComputerNetwork08_energy}\cite{Secon05_energy}\cite{TWC05_energy}, and etc.

In this paper, we study the problem of \emph{flow demands oriented node placement}, i.e., how to use less wireless nodes to satisfy flow requirements for multi-hop wireless networks. Here, flow requirements indicate that the data should be transmitted at a certain level of transmission rate. Existing studies on node placement related to network traffic \cite{SAC13_throughput}\cite{MNA08_throughput}\cite{JSAC12_throughput} mainly address the problem of optimizing the network throughput. However, the flow demands has significant difference from the throughput demands and new methods are needed to satisfy flow demands for node placement. Actually, it is not straightforward to solve the problem of flow demands oriented node placement. Therefore, we divide and conquer the problem by three steps: at first, we calculate the maximal flow for single routing demand, which is the basis of analyzing multiple routing demand; then, we calculate the maximal flow for multiple routing demands according to the result on single routing demand; finally, based on the maximal flow calculated for multiple routing demand, we try to merge routing paths to achieve minimal number of wireless nodes. The above procedure involves several proofs and related algorithm design, i.e., two problems are proved as NP-hard and two heuristic algorithms are proposed. In evaluation, we verify the efficiency of the proposed algorithms by examining \emph{average satisfied rate} (defined in Section IV) of flow demands and the number of nodes used for placement under the practical scenarios of data aggregation, demands with definite flow requirement and nodes with unknown flow requirement.

To the best of our knowledge, the problem of satisfying the flow demands by node placement in multi-hop wireless networks has not been studied yet. The contributions are summarized as follows:

\begin{enumerate}
  \item \emph{The complexity of flow demands oriented node placement is theoretically analyzed.} For single routing demand, the theoretical maximal flow can be achieved under the interference model is conducted. For multiple routing demands, the proof of NP-hard to obtain the maximal flow under the interference between routing paths, and the proof of NP-hard to minimize the number of nodes placed by merging routing paths are given.
  \item \emph{A novel approach is proposed to use less nodes to satisfy the flow requirements for node placement.} For multiple routing demands, a polynomial-time complexity algorithm is given to achieve larger flow as possible by finding the heaviest interference node and assigning activated time slots on each routing path, and a polynomial-time complexity algorithm is given to place relay nodes as fewer as possible by prior to merge the constructed routing paths with more excessive flow capacity. Both proposed algorithms are discussed with their worst case.
  \item \emph{The efficiency of the proposed algorithms are verified through several practical scenarios.} Our experiments show that Average Satisfied Rate of flow demands reduces slowly when increasing the level of flow requirements, and the number of nodes for placement is reduced in average of 25.4\%, 26.6\% and 24.1\%, both of which prove the efficiency of proposed algorithms.
\end{enumerate}

The rest of the paper is organized as follows. Section \ref{model} introduces the model assumption. Section \ref{node_placement} studies the flow demands oriented node placement problem and proposes the algorithm. Section \ref{experiment} evaluates the proposed algorithm. Section \ref{conclusion} concludes the work.

\section{The Problem and Related Work} \label{model}
Suppose that a set of nodes have been deployed in a plane. Some of these nodes may have routing demands of transmitting their data to other nodes with flow requirements that should be achieved. Then, relay nodes are needed to be placed to route and satisfy these flow demands. We assume that each node is equipped with a radio, and the radio has the its maximal transmission range $r$. If two nodes, $n_1$ and $n_2$, that are located with their distance $Dist(n_1,n_2)\leq R$ can interfere the transmission between each other. With respect to a successful transmission of two nodes, we consider the \emph{Protocol Interference Model}. A transmission from node $n_1$ to $n_2$ is successful if and only if 1) $Dist(n_1,n_2) \leq r$; 2) there does not exist a transmission node $n_3$ such that $Dist(n_1,n_3) \leq R$. This interference model is widely used in references like \cite{TIT00_capacity}\cite{Mobi03_Interference}. Besides, we use a time slotted system. In the time slotted system, the time is divided into equal length slots, and the transmission between nodes are synchronized. We define $f$ as the maximal flow (or maximal transmission rate) that can be transmitted in a single time slot.

For the set of nodes with flow requirements, we model that there are total $m$ pair of routing demands $(src_q, dest_q)$ with flow requirement $F_A(src_qdest_q)$ should be satisfied, $q=1,2,...,m$. In this paper, we aim to place relay nodes to construct paths to satisfy the flow requirement $F_A(src_qdest_q)$ between $(src_q, dest_q)$, $q=1,2,...,m$, while use as fewer placed nodes as possible. We named it as the flow demands oriented node placement problem for short.

As a matter of fact, the above problem has not been studied in previous research. Some of existing studies \cite{SAC13_throughput}\cite{MNA08_throughput}\cite{JSAC12_throughput} on node placement address the throughput issue. These studies mainly focus on placing nodes to optimize the network throughput. The scenario of collecting all node transmissions towards sink nodes is considered in \cite{SAC13_throughput}. A grid-based relay nodes placement method is used to optimize the network throughput in  \cite{MNA08_throughput}. Two objectives are studied in \cite{JSAC12_throughput}. The one is to maximize the minimum throughput for any relay node, and the other is to maximize the total throughput of the network. However, the throughput optimization of the whole network is not equal to the issue of satisfying the flow requirements of some nodes, and also, the flow requirements of different nodes may be different, which should be with differential treatment. The flow requirements cannot be satisfied within throughput optimization framework. Besides, the wireless interference, an intrinsic characteristic of wireless channel, among nodes is also seldom considered in these studies. We study the flow demands oriented node placement problem to fill this blank.



\section{Flow Demands Oriented Node Placement} \label{node_placement}

\subsection{Methodology}

In this section, we solve the problem of flow demands oriented node placement in following three steps:
\begin{enumerate}
  \item \emph{Calculating the maximal flow for single routing demand.} At first, we find the maximal flow can be achieved by constructing paths for a single routing demand $(src_1, dest_1)$.
  \item \emph{Calculating the maximal flow for multiple routing demands.} Then, we prove that constructing paths to maximize the flow for multiple routing demands $(src_q, dest_q)$, $q=1,2,...,m$ is NP-hard, and try to design an algorithm to achieve greater flow as possible for multiple routing demands under the interference between paths.
  \item \emph{Finding the minimal number of nodes for multiple routing demands with the flow requirement. }Finally, we also prove the problem of reducing the number of nodes placed to the minimal with flow requirement $F_A(src_qdest_q)$ for $(src_q, dest_q)$, $q=1,2,...,m$ is NP-hard, and try to propose our algorithm to reduce the number of nodes placed as possible while satisfying the flow requirement.
\end{enumerate}

\subsection{Calculating The Maximal Flow for Single Routing Demand}

Suppose that there is a node pair $(src_1,dest_1)$ with routing demand from source node $src_1$ to destination node $dest_1$. Let $p_{src_1dest_1,1}$ be the path constructed between $src_1$ and $dest_1$, and $f_{src_1dest_1}$ be the flow between $src_1$ and $dest_1$. Assume that the interference range $R$ is in $[jr,(j+1)r)$, $j \in \mathbb{N^+}$, and the distance $Dist(src_1,dest_1)$ between $src_1$ and $dest_1$ has $Dist(src_1,dest_1)\geq 2R$\footnote{Node $src_1$ and $dest_1$ can directly communicate or communicate with only a few of relay nodes if we have $Dist(src_1,dest_1)<2R$, so it is meaningless to study the node placement problem under the condition of $Dist(src_1,dest_1)\leq 2R$. Also, it will be involved in complex discussion of cases under this condition}. We give a theorem as follows:

\begin{thm} \label{t1}
With a single path constructed, the flow $f_{src_1dest_1}$ between $src_1$ and $dest_1$ can reach $F_1 = \frac{f}{j+1}$ at most.
\end{thm}
\begin{proof}
Assume that the distance between $src_1$ and $dest_1$ has $Dist(src_1,dest_1) \in (ir,(i+1)r]$. At least $i$ nodes should be placed to construct a path $p_{src_1dest_1,1}$ between $src_1$ and $dest_1$ as shown in figure \ref{fig1}.
\begin{figure}[!htb]
\centering
\includegraphics[width=3.0in]{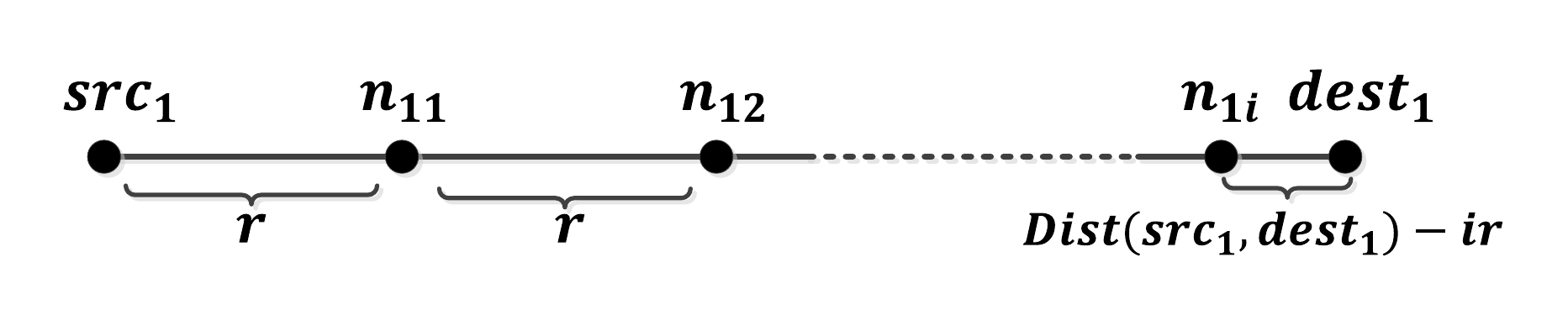}\\
\caption{\textrm{A single path of $(src_1,dest_1)$}} \label{fig1}
\end{figure}
With the condition $R \in [jr,(j+1)r)$, $j \in \mathbb{N^+}$, we know that node $n_{11}$, $n_{12}$, ..., $n_{1j}$ are within interference range of $src_1$. The total $j+1$ links of $src_1n_{11}$, $n_{11}n_{12}$, ..., $n_{1j-1}n_{1j}$ should be assigned with different time slots to avoid mutual interference in transmission. The other links with transmission nodes out of the interference range of $src_1$ can be assigned with the time slots that have been used, i.e. $Slot(n_{1j+1}n_{1j+2})=Slot(src_1n_{11})=1$, $Slot(n_{1j+2}n_{1j+3})=Slot(n_{11}n_{12})=2$, ..., $Slot(n_{1i-1}n_{1i})=(i-1)\%(j+1)+1$ and $Slot(n_{1i}dest_1)=i\%(j+1)+1$. There are total $j+1$ time slots assigned to the path $p_{src_1dest_1,1}$.
\end{proof}

As proved, the flow $f_{src_1dest_1}$ between $src_1$ and $dest_1$ can reach $F_1$ at most with one path $p_{src_1dest_1,1}$ constructed. With multiple paths constructed, the flow $f_{src_1dest_1}$ can reach a larger value. We have the following theorem:

\begin{thm} \label{thm_Fc}
With multiple paths constructed, the flow $f_{src_1dest_1}$ between $src_1$ and $dest_1$ can reach $F_C$ at most.
\begin{equation} \label{h1}
F_C = \max_{c=1,2,...,C} \frac{cf}{s_c}
\end{equation}
Where $s_c = \max \limits_{q=1,2,...,j} \{ \; q+1+\sum_{m=2}^c \{ \; \max \{ \; floor(\frac{x_m}{r}), \; 0 \} + 1 \; \} \; \}$ and $x_m =\sqrt{q^2r^2 \cos^2{\frac{2 \pi (m-1)}{c}} - q^2r^2 + R^2} + qr\cos{\frac{2 \pi (m-1)}{c}}$.
\end{thm}
\begin{proof}
Suppose that there are $c$ paths constructed between $src_1$ and $dest_1$. To find the maximal flow that can be achieved by the constructed paths, the key is to seek out the heaviest interference area that involves the maximum number of nodes in transmission. Obviously, the area surrounding $src_1$ or $dest_1$ are inevitable to be the heaviest interference area with the highest density of nodes placed in the construction of paths, i.e. some nodes have to be placed around $src_1$ to receive the transmission from $src_1$ in figure \ref{fig3}, and thus, these nodes could not be placed far enough to avoid interference. To find the maximal flow can be achieved in the heaviest interference area surrounding $src_1$ or $dest_1$, it is necessary to find the heaviest interference node on a constructed path, which is the node that interferes with the maximum number of nodes on other paths. Let $S_{int}(n)$ denote the interference node set for a node $n$. Take the heaviest interference area surrounding $src_1$ for example. The heaviest interference node $n_h$ has
\begin{equation} \label{n1}
n_h = \arg \max_{Dist(src_1,n)<R} |S_{int}(n)|
\end{equation}

\begin{figure}[!htb]
\centering
\includegraphics[width=2.9in]{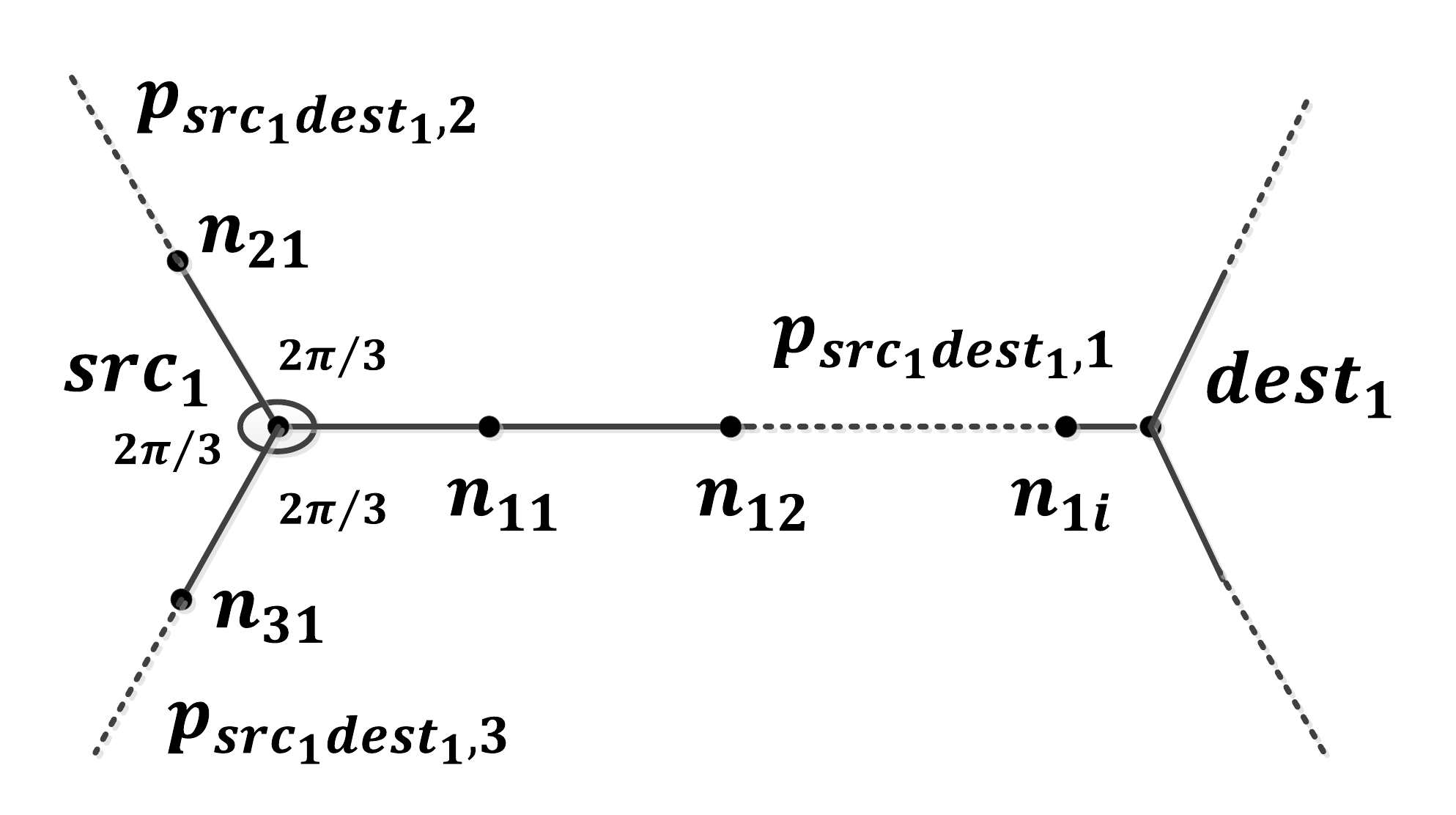}\\
\caption{\textrm{$3$ constructed paths of $(src_1,dest_1)$}} \label{fig3}
\end{figure}

Next, we firstly try to minimize the interferes when constructing paths so that the heaviest interference node will interfere with the minimum number of nodes. With the path construction process, the expression (\ref{n1}) can be rewritten as
\begin{equation} \label{n2}
n_h = \arg \min \{ \max_{Dist(src_1,n)<R} |S_{int}(n)| \; | \; path \; construction \}
\end{equation}
It can be proved that if these $c$ paths are constructed with equal angle-interval of $\frac{2\pi}{c}$ from $src_1$ and to $dest_1$ ($\angle n_{11}src_1n_{21}=\angle n_{21}src_1n_{31}=\angle n_{31}src_1n_{11}=\frac{2\pi}{3}$ in figure \ref{fig3}), the interference can be reduced to the minimal between paths. It is not hard to finish this proof with the basic knowledge of analytic geometry. As space is limited, we ignore the subordinate details here.

After $c$ paths with equal angle-interval of $\frac{2\pi}{c}$ have been constructed, the heaviest interference node $n_h$ could be found by expression \ref{n1}. More specific, we randomly select a constructed path and compare all the nodes $n$ having $Dist(src_1,n)<R$ to find the heaviest interference node that determines the maximal flow can be achieved. Without loss of generality, we set the selected path as $p_{src_1dest_1,1}$ and the selected comparison nodes with the coordinates of $(r,0),(2r,0),...,(jr,0)$. For a node $(qr,0)$, $q=1,2,...,j$, the interference length $x_m$ with another path $p_{src_1dest_1,m}$ can be calculated with the cosine formula
\begin{equation} \label{n3}
R^2 = q^2r^2 + {x_m}^2 - 2qrx_m \cos{\frac{2 \pi (m-1)}{c}}
\end{equation}
Solve this equation, we can get
\begin{equation} \label{n4}
x_m = qr\cos{\frac{2 \pi (m-1)}{c}} + \sqrt{q^2r^2 \cos^2{\frac{2 \pi (m-1)}{c}} - q^2r^2 + R^2}
\end{equation}
Since $x_m$ could be negative value, we can calculate the number of interference nodes on path $p_{src_1dest_1,m}$ as $\max\{floor(x_m/r), 0\}$, in which $floor()$ rounds down the fractions. Thus, $s_m$ new time slots should be assigned to $p_{src_1dest_1,m}$
\begin{equation} \label{n5}
s_m = \max\{\; floor(\frac{x_m}{r}), \; 0 \; \} + 1
\end{equation}
Add up all the new assigned time slots
\begin{equation} \label{n6}
s(qr,0) = q+1+\sum_{m=2}^c s_m
\end{equation}
Find the maximal counts of time slots assigned
\begin{equation} \label{n7}
s_c = \max_{q=1,2,...,j} s(qr,0)
\end{equation}
The node that introduces the maximal counts of time slots $s_c$ is the heaviest interference node. We get the maximal counts of time slots $s_c$ for $c$ constructed paths, and the flow achieved can be denoted as
\begin{equation} \label{n8}
f_c = \frac{cf}{s_c}
\end{equation}
Let $c=1,2,...,C$. We can obtain the maximal flow $F_C$ can be achieved between $src_1$ and $dest_1$ by\footnote{The value of $C$ is determined by $R$. When the number of paths increase to a certain value, the flow can be achieved will decrease then with too many nodes interfering in transmission.}
\begin{equation} \label{n9}
F_C = \max_{c=1,2,...,C} f_c
\end{equation}
The proof is completed by combining the expression (\ref{n4})-(\ref{n9}).
\end{proof}

\begin{algorithm} \label{alg1}
\caption{Construct Multiple Paths for One Source Destination Pair (MP1)}
\KwIn{Source destination pair $(src_1,dest_1)$, Flow $f$, Interference range $R$, Maximal path count $C$}
\KwOut{The maximal flow $F_C$ and its constructed paths $p_{src_1dest_1}$}
$F_C=0$; \\
\For{$c=1,2,...,C$}{
    \For{$m=1,2,...,c$}{
        Place nodes as the beginning of path $p_{src_1dest_1,m}$ with $\frac{2\pi(m-1)}{c}$ degrees to the direction of $\overrightarrow{src_1dest_1}$ from $src_1$ until a placed node on $p_{src_1dest_1,m}$ does not interfere other nodes in other already constructed paths; \\
        Place nodes as the end of path $p_{src_1dest_1,m}$ with $\pi - \frac{2\pi(m-1)}{c}$ degrees to the direction of $\overrightarrow{src_1dest_1}$ to $dest_1$ until a placed node on $p_{src_1dest_1,m}$ does not interfere other nodes in other already constructed paths; \\
        Connect the nodes at the beginning and the end of $p_{src_1dest_1,m}$ and let the nodes at the connection path do not interfere the transmission of other already constructed paths; \\
    }
    $s_c = \max \limits_{q=1,2,...,j} \{ \; q+1+\sum_{m=2}^c \{ \; \max \{ \; floor(\frac{x_m}{r}), \; 0 \} + 1 \; \} \; \}$; \\
    $x_m =\sqrt{q^2r^2 \cos^2{\frac{2 \pi (m-1)}{c}} - q^2r^2 + R^2} + qr\cos{\frac{2 \pi (m-1)}{c}}$; \\
    \If{$\max_{c=1,2,...,C} \frac{cf}{s_c}>F_C$}{
        $F_C=\max_{c=1,2,...,C} \frac{cf}{s_c}$; \\
        Record the constructed paths $p_{src_1dest_1,m}$, $m=1,2,...,c$; \\
    }
}
\Return $F_C$ and $p_{src_1dest_1}$; \\
\end{algorithm}

We have proved that the flow $f_{src_1dest_1}$ between $src_1$ and $dest_1$ can reach $F_C$ at most with multiple paths used. The One Source Destination Pair with Multiple Paths (MP1) algorithm summarizes the process to find $F_C$ (Algorithm \ref{alg1}). Paths are constructed with equal angle-interval at their beginning and end part, and the connection paths between the beginning and end part are constructed to avoid interference mutually. Then, the maximal flow $F_C$ can be calculated by expression (\ref{h1}). MP1 algorithm executes with time complexity of $O(C^2)$ to construct paths.

\subsection{Calculating The Maximal Flow for Multiple Routing Demands}

Suppose that there are $m$ routing demands $(src_1,dest_1)$, $(src_2,dest_2)$,...,$(src_m,dest_m)$. As proved, the flow of a single demand $(src_1,dest_1)$ can reach the flow of $F_C$ by constructing multiple paths. However, in $m$ routing demands case, not all these demands can reach $F_C$. We have the following theorem:
\begin{thm} \label{t3}
Not all $m$ routing demands of node pair $(src_1,dest_1)$, $(src_2,dest_2)$,...,$(src_m,dest_m)$ can definitely reach the flow of $F_C$ by constructing paths.
\end{thm}

For example, there are two routing demands $(src_1,dest_1)$ and $(src_2,dest_2)$ in figure \ref{fig4}. Both $src_1$ and $dest_1$ interfere with $src_2$ in transmission. When $src_2$ does not transmit, the flow of $(src_1,dest_1)$ can reach $F_c$. When $src_2$ transmits with interference, the link $src_1dest_1$ cannot be activated all the time, and thus, cannot reach the flow of $F_C$.

\begin{figure}[!htb]
\centering
\includegraphics[width=1.9in]{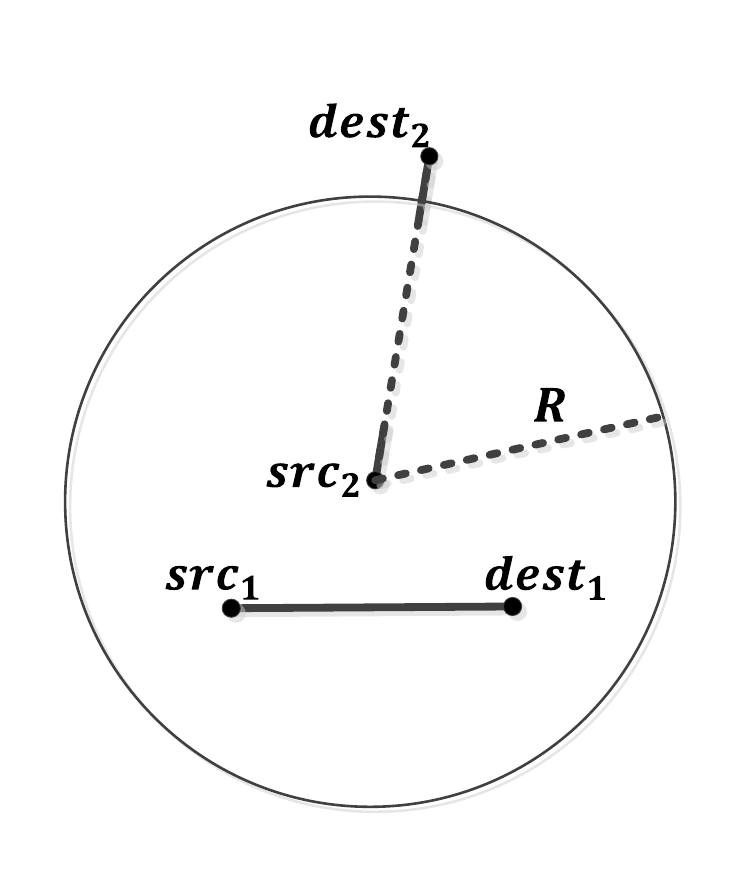}\\
\caption{\textrm{Two routing demands $(src_1,dest_1)$ and $(src_2,dest_2)$}} \label{fig4}
\end{figure}

\begin{figure}[!htb]
\centering
\includegraphics[width=2.8in]{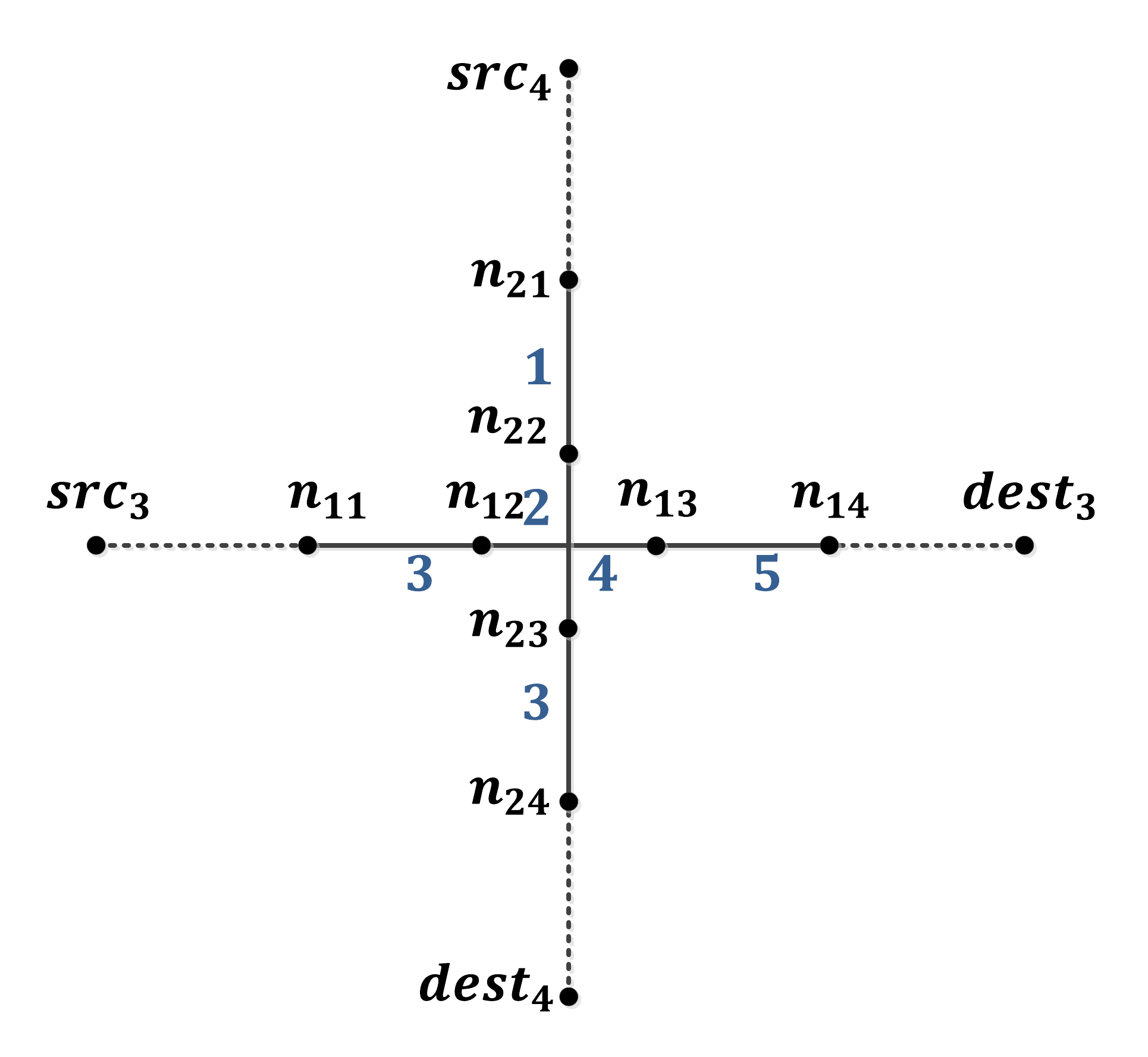}\\
\caption{\textrm{Two routing demands $(src_3,dest_3)$ and $(src_4,dest_4)$}} \label{fig5}
\end{figure}

However, even when the transmission of a routing demand is under the interference of constructed paths of other demands, the flow of this routing demand can still reach $F_C$ by assigning time slots in some cases. In figure \ref{fig5}, there are two routing demands $(src_3,dest_3)$ and $(src_4,dest_4)$. The path $p_{src_3dest_3,1}$ is assigned with time slots set of $\{1,2,3,4,5\}$ to its links, and thus, can reach flow of $\frac{f}{5}$. The path $p_{src_4dest_4,2}$ is assigned with time slots $\{1,2,3,4\}$ and can reach flow of $\frac{  }{4}$. Both node $n_{12}$ and $n_{13}$ on path $p_{src_3dest_3,1}$ interfere with $n_{22}$ and $n_{23}$ in transmission. By assigning time slots of $4$, $5$, $2$ and $3$ to links $n_{12}n_{13}$, $n_{13}n_{14}$, $n_{22}n_{23}$ and $n_{23}n_{24}$ respectively, both path $p_{src_3dest_3,1}$ and $p_{src_4dest_4,2}$ can still reach the flow of $\frac{f}{5}$ and $\frac{f}{4}$ without adding new slot to existing time slots set to avoid interference.

\begin{algorithm} \label{alg2}
\caption{Construct Multiple Paths for Multiple Source Destination Pairs (MPM)}
\KwIn{Source destination pairs $(src_q,dest_q)$, $q=1,2,...,m$, Flow $f$, Interference range $R$, Maximal path count $C$}
\KwOut{The flow $F^{'}_C(src_qdest_q)$ and path $p_{src_qdest_q}$, $q=1,2,...,m$}

\For{$q=1,2,...,m$}{
    Get the pair $(src_q,dest_q)$; \\
    $[F_C(src_qdest_q), p_{src_qdest_q}]=MP1((src_q,dest_q), f, R, C)$; \\
}
\For{$q=1,2,...,m$}{
    Get the pair $(src_q,dest_q)$; \\
    Set $F^{'}_C(src_qdest_q)=0$; \\
    \For{Each constructed path $p$ in $p_{src_qdest_q}$}{
        Set $n(p)$ as a randomly selected node on path $p$; \\
        Set $S_{int}(p)=\emptyset$; \\
        \For{Each placed node $n$ in $p$}{
            Find the nodes that interfere the transmission of node $n$, and record these nodes into $S_{int}(n)$ (Count the source and destination nodes with the times of the number of constructed paths that connect them in $S_{int}(n)$\footnote{For example, if there are $3$ paths constructed from $src_1$ to $dest_1$, and $src_1$ or $dest_1$ interferes the transmission of node $n$, then $src_1$ or $dest_1$ is counted for $3$ times in $S_{int}(n)$.}) \\
            \If{$|S_{int}(n)|>|S_{int}(p)|$}{
                Set $n(p)=n$; \\
                Set $S_{int}(p)=S_{int}(n)$; \\
            }
        }
        Assign time slots ${1,2,...,|S_{int}(p)|}$ to the links transmitted by the elements in $S_{int}(p)$ if the link has not been assigned with time slot, and other links in path $p$ could reuse these time slots without interference in transmission; \\
        Set $F^{'}_C(src_qdest_q)+=\frac{1}{S_{int}(p)}$; \\
    }
}
\Return The flow $F^{'}_C(src_qdest_q)$ and path $p_{src_qdest_q}$, $q=1,2,...,m$; \\
\end{algorithm}

As a summary, we give a theorem:
\begin{thm}
The flow of a routing demand can reach $F_C$ if time slots can be assigned to avoid interference of other demands without adding new time slots.
\end{thm}

Since not every routing demand can reach its maximal flow of $F_C$, we consider $m$ routing demands as a whole, and try to answer the problem that how to maximize the sum of the flow of all these $m$ routing demands. The problem is formalized as following objective function:
\begin{equation} \label{e4}
    \begin{array}{rcl}
        \max {\sum_{q=1}^{m} f_{src_qdest_q}}\\
        \\
        s.t. \quad f_{src_qdest_q}>0
    \end{array}
\end{equation}
However, it is NP-hard to solve this objective function.
\begin{thm} \label{t4}
Construct paths to maximize the sum of the flow of $m$ routing demands is NP-hard.
\end{thm}
\begin{proof}
This problem can be reduced to the NP-hard problem of Theorem $1$ in \cite{Mobi03_Interference}. The NP-hard problem in \cite{Mobi03_Interference} is to find the maximal sum of the flow for a group of source and destination nodes on a given network with deployed nodes. Our problem is difference in that nodes are placed to construct paths to maximize the sum of the flow for a group of source and destination nodes. The node placement can be seen as selecting a set of nodes from infinite candidate nodes from the plane. By considering these candidate nodes as deployed nodes in the network, our problem can be reduced to the NP-hard problem of Theorem $1$ in \cite{Mobi03_Interference}.
\end{proof}

It is NP-hard to maximize the sum of the flow of $m$ routing demands by constructing paths. Also, it has been proved in \cite{Mobi03_Interference} that the maximum sum of the flow is NP-hard to be approximated. In other words, with the reduction of our problem to this problem, we know that it is NP-hard to find a solution to construct paths that guarantee to approximate the maximal flow.

Although it is NP-hard to approximate the maximal flow, the constructed paths should try to avoid interference between each other. The interference between paths depends on the position of the source and destination nodes. Reconstruct the paths if interference exists between them can make little effect because reconstructed paths usually introduce new interference, especially when the reconstructed path goes through the source or destination nodes, or intersects with other paths. So we design the MPM algorithm that does not try to reconstruct paths to avoid interference, but try to find the heaviest interference node and assign time slots based on its interference set. MPM algorithm (Algorithm \ref{alg2}) find the heaviest interference node $n(p)$ for path $p$, and record its interference nodes to set $S_{int}(p)$. The flow $F^{'}_C(p)$ achieved by path $p$ is $\frac{1}{|S_{int}(p)|}$. The flow $F^{'}_C(src_qdest_q)$ between $src_q$ and $dest_q$ can be calculated by summing up the flow of constructed paths between them. MPM algorithm processes each path once, and returns the flow $F^{'}_C(src_qdest_q)$ and constructed path $p_{src_qdest_q}$, $q=1,2,...,m$. The achieved sum of flow $F^{'}_C$ for all source and destination pairs is $\sum_{q=1}^m F^{'}_C(src_qdest_q)$.

\subsection{Finding The Minimal Number of Nodes for Routing Demands with Flow Requirement}

\begin{algorithm} \label{alg3}
\caption{Merge}
\KwIn{Paths $p_{src_qdest_q}$ with flow requirement $F_A(src_qdest_q)$, $q=1,2,...,m$, Flow $F^{'}_C$}
\KwOut{Merged paths $p^{'}_{src_qdest_q}$, $q=1,2,...,m$}
\For{$q=1,2,...,m$}{
    Let $c_q$ denote the number of paths for $(src_q, dest_q)$; \\
    \While{Delete the longest path from $p_{src_qdest_q}$, still having the flow achieved $F^{'}_{c_q}(src_qdest_q)>F_A(src_qdest_q)$}{
        Delete the longest path from $p_{src_qdest_q}$; \\
        Set $c_q=c_q-1$; \\
    }
    Set $dF(src_qdest_q)=F^{'}_{c_q}(src_qdest_q)-F_A(src_qdest_q)$; \\
}
Sort the $(src_q, dest_q)$ by the value of $dF(src_qdest_q)$ in descending order, $q=1,2,...,m$, and record the result in array $D$; \\
\For{$q=1,2,...,m$}{
    \For{$u=q+1,q+2,...,m$}{
        Get the pair $(src_q, dest_q)$ and $(src_u, dest_u)$ from array $D$; \\
        Equal-lengthly sample a set of points in all paths of $p_{src_qdest_q}$ and $p_{src_udest_u}$; \\
        Calculate the distance from all paths of $p_{src_qdest_q}$ to all paths of $p_{src_udest_u}$ by averaging the distance between sample points, and record the results into the $c_q\times c_u$ matrix $M$; \\
        \While{$1$}{
            Select the minimal value of distance from matrix $M$, denoted by $M(a_q, a_u)$ ($M(a_q, a_u)\neq \infty$); \\
            Set the $a_q$ row $M(a_q,:)=\infty$, and the $a_u$ column $M(:,a_u)=\infty$; \\
            \eIf{$p_{src_qdest_q,a_q}$ and $p_{src_qdest_q,a_q}$ is mergable}{
                Merge the path $p_{src_qdest_q,a_q}$ and $p_{src_qdest_q,a_q}$; \\
            }{
                break;
            }
        }
        Record the newly constructed path $p^{'}_{src_qdest_q}$ and $p^{'}_{src_udest_u}$; \\
    }
}
\Return paths $p^{'}_{src_qdest_q}$, $q=1,2,...,m$;
\end{algorithm}

$m$ routing demands with flow requirement of $F_C$ has been discussed in last section. Using MPM algorithm, $m$ routing demands can achieve flow of $F^{'}_C$. Actually, the flow requirement of routing demand can be arbitrary in practical. In this section, we focus on the case of $m$ routing demands $(src_1,dest_1)$, $(src_2,dest_2)$,...,$(src_m,dest_m)$ with flow requirement of  $F_A(src_qdest_q)$, in which $F_A(src_qdest_q)\leq F^{'}_C(src_qdest_q)$ for $q=1,2,...,m$. In this case, the capacity of some paths may exceed the flow requirement of some routing demands. This will lower link utilization rate of these paths. It is a waste of resource. For this reason, some paths can be merged to increase the link utilization rate and decrease the total number of nodes needed to be placed. We have the following objective function:
\begin{equation} \label{e5}
    \begin{array}{c}
        \min {\sum_{q=1}^{m} \sum_{a=1}^{c_q} Length(p_{src_qdest_q,a})}  \\
        \\
        s.t. \quad \sum_{a=1}^{c_q} f_{src_qdest_q,a} \geq F_A(src_qdest_q)
    \end{array}
\end{equation}

As implied in the objective function, two path, one for routing demand $(src_q,dest_q)$ and the other one for $(src_{q_k},dest_{q_k})$, can be merged if the length of the merged path is reduced than the total length of these two paths, and the flow requirement, $F_A(src_qdest_q)$ and $F_A(src_{q_k}dest_{q_k})$, can be met after these two path have been merged. I.e., in figure \ref{fig6}, let $Length(d_{k1}d_{k2})$ denote the length of path $p_{src_{q_k}dest_{q_k},a_k}$ between two points, $d_{k1}$ and $d_{k2}$. We have $Dist(d_{k1},p_{src_qdest_q,a}) + Dist(d_{k1},p_{src_qdest_q,a}) < Length(d_{k1}d_{k2})$, which makes the merged path shorter than the total length of the original two paths and the flow requirement of $F_A(src_qdest_q)$ and $F_A(src_{q_k}dest_{q_k})$ can still be met. The path $p_{src_qdest_q,a}$ and $p_{src_{q_k}dest_{q_k},a_k}$ are mergable. We use the mergable element $P$ to denote mergable paths, such as $P=(p_{src_qdest_q,a}, p_{src_{q_k}dest_{q_k},a_k})$. A mergable set $S_{mer}$ composes of mergable elements, having $p_1\neq p_2$ for $\forall p_1\in P_1$, $\forall p_2\in P_2$, and $\forall P_1, P_2 \in S_{mer}$. Then, a maximum mergable set can be described as there are no other mergable paths for $(src_q,dest_q)$, $q=1,2,...,m$ after applying all the merging for mergable paths in the set.

\begin{figure} \centering
\subfigure[Before merging] { \label{fig6:a}
\raggedleft
\includegraphics[width=2.4in]{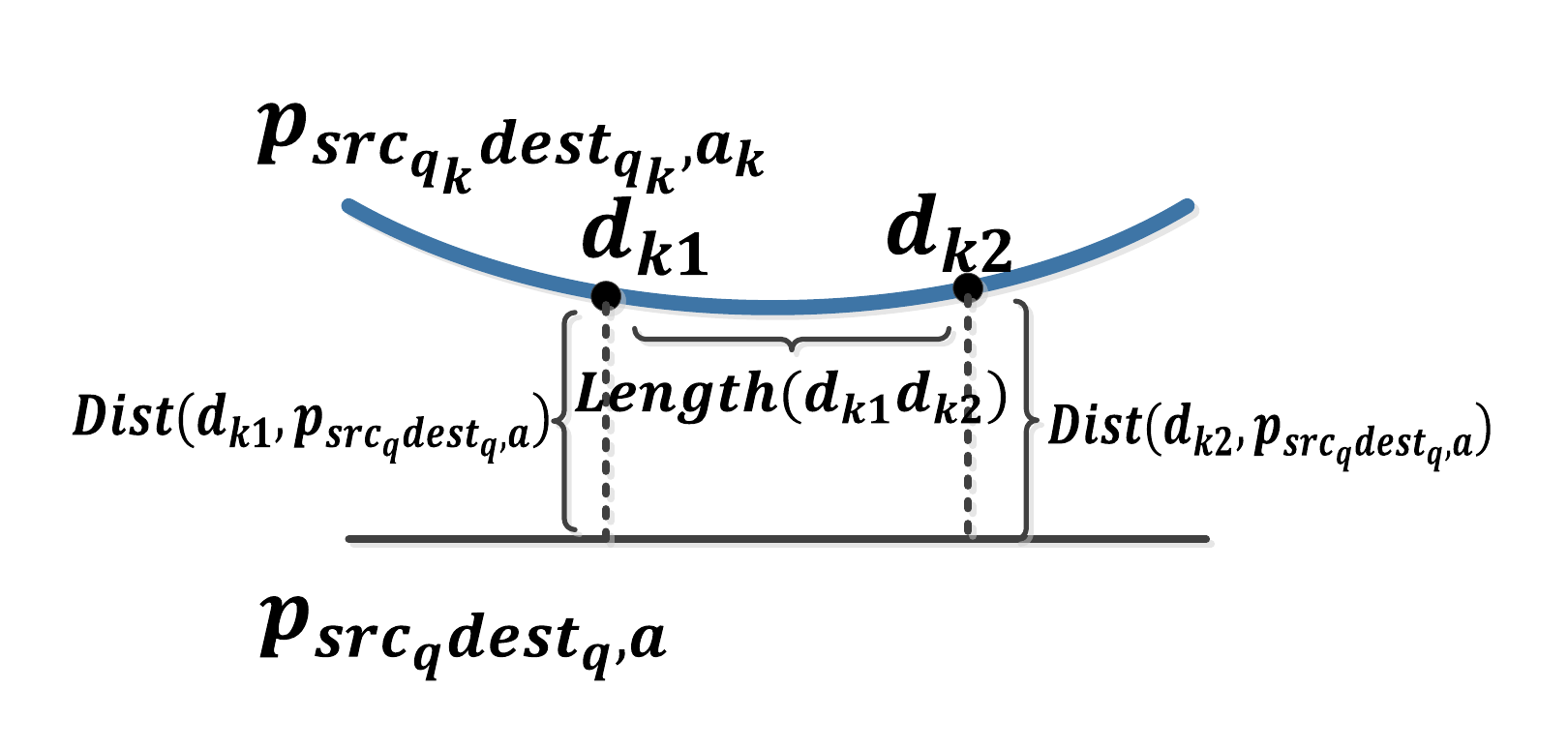}
}
\subfigure[After merging] { \label{fig6:b}
\centering
\includegraphics[width=1.6in]{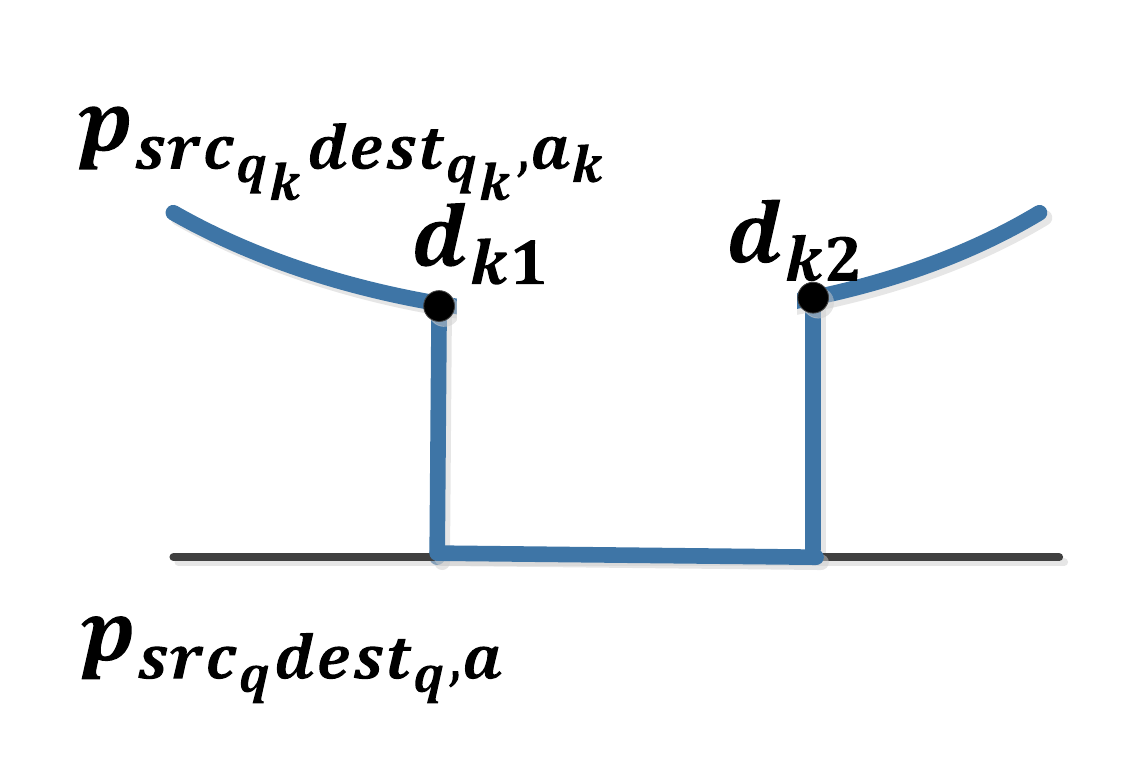}
}
\caption{Merge path}
\label{fig6}
\end{figure}

To solve the objective function that minimize the total path length while still meet the flow demands, all the maximum mergable set should be found and compared so that we can find the set minimize the total path length. However, it is NP-hard to solve this objective function.

\begin{thm}
Minimize the total length of all paths for $m$ routing demands with flow requirement restriction by merging paths is NP-hard.
\end{thm}
\begin{proof}
Let each mergable element represent a node in the network. Establish a link between nodes if two mergable elements contain a same path. Then, our problem can be reduced to the NP-hard problem of finding all maximal independent set in the network \cite{MIS88}.
\end{proof}

It is NP-hard to minimize the total length of all paths for $m$ routing demands by merging paths. The Merge algorithm (Algorithm \ref{alg3}) is proposed to try to minimize the total length of constructed paths while still meet the flow requirement $F_A$ of these $m$ routing demands. In the Merge algorithm, extra paths with longer length of routing demands are firstly deleted, and the extra flows $dF(src_qdest_q)$, $q=1,2,...,m$ that represents the excessive flow compared to flow demand $F_A(src_qdest_q)$ are recorded. The paths of a routing pair with more excess flow are more likely to be merged with other paths, so we sort the routing pairs by $dF(src_qdest_q)$, $q=1,2,...,m$ in descending order, and try to merge the paths with more excess flow (detailed in Algorithm \ref{alg3}). Finally, the Merge algorithm returns the newly merge path $p^{'}_{src_qdest_q}$, $q=1,2,...,m$ with $O(m^2)$ merge attempts between routing demands. The merging between paths depends on the positions of $(src_q, dest,q)$, $q=1,2,...,m$, and their flow requirement $F_A(src_qdest_q)$. In worst case, no path can be merged to reduce the total path length.

\subsection{Summary}
In this section, we first analyze a single routing demand of $(src_1,dest_1)$ and prove that the flow can reach $F_1$, $F_C$ with one path and multiple paths constructed respectively. Then, for multiple routing demands of $(src_q,dest_q)$, $q=1,2,...,m$, we prove that it is NP-hard to maximize the sum of the flow by constructing paths and show that the flow can achieve $F^{'}_C$ by MPM algorithm. For multiple routing demands with flow requirement, we prove that it is NP-hard to minimize the total length of all paths by merging path and propose a greedy algorithm of merging paths.

\begin{algorithm} \label{alg4}
\caption{Flow Demands Oriented Node Placement}
\KwIn{Source destination pairs $(src_q,dest_q)$ with flow requirement $F_A(src_qdest_q)$, $q=1,2,...,m$, Flow $f$, Interference range $R$, Maximal path count $C$}
\KwOut{The placement position $P$ of nodes}
$[F^{'}_C, p_{src_qdest_q}]=MPM((src_q,dest_q),f,R,C)$, $q=1,2,...,m$;\\
$p^{'}_{src_qdest_q}=Merge(F^{'}_C, F_A, p_{src_qdest_q})$, $q=1,2,...,m$;\\
Place nodes along the constructed paths $p^{'}_{src_qdest_q}$;\\
Record the position of nodes to $P$;\\
\Return $P$
\end{algorithm}

The Flow Demands Oriented Node Placement Algorithm is summarized in Algorithm \ref{alg4}. In line 1, the MPM algorithm (Algorithm \ref{alg2}) constructs paths to find the maximal flow of multiple routing demands that can be supported by the network. The MPM algorithm try to satisfy the flow $F_C$ obtained by the MP1 algorithm (Algorithm \ref{alg1}) and can achieve the flow of $F^{'}_C$. In line 2, given flow requirement $F_A$ of routing demands, The Merge algorithm (Algorithm \ref{alg3}) try to minimize the total length of paths. Then, the algorithm places nodes along the merge paths and returns the position of placed nodes.

Following the analysis above, proposed Algorithm \ref{alg4} processes routing demands in $O(m^2)$ times. The performance of the MPM algorithm depends on the positions of $(src_q, dest,q)$, $q=1,2,...,m$. In worst case, only one link can be activated for transmission in each time slot. The performance of the Merge algorithm depends on the positions of $(src_q, dest,q)$, $q=1,2,...,m$, and their flow requirement $F_A(src_qdest_q)$. In worst case, no path can be merged at all. We will examine the performance of proposed algorithm later in evaluation.

\section{Evaluation} \label{experiment}

\begin{figure*}[!htb]
\begin{minipage}[]{0.33\textwidth}
\centering\includegraphics[height=1.41in,width=1.9in]{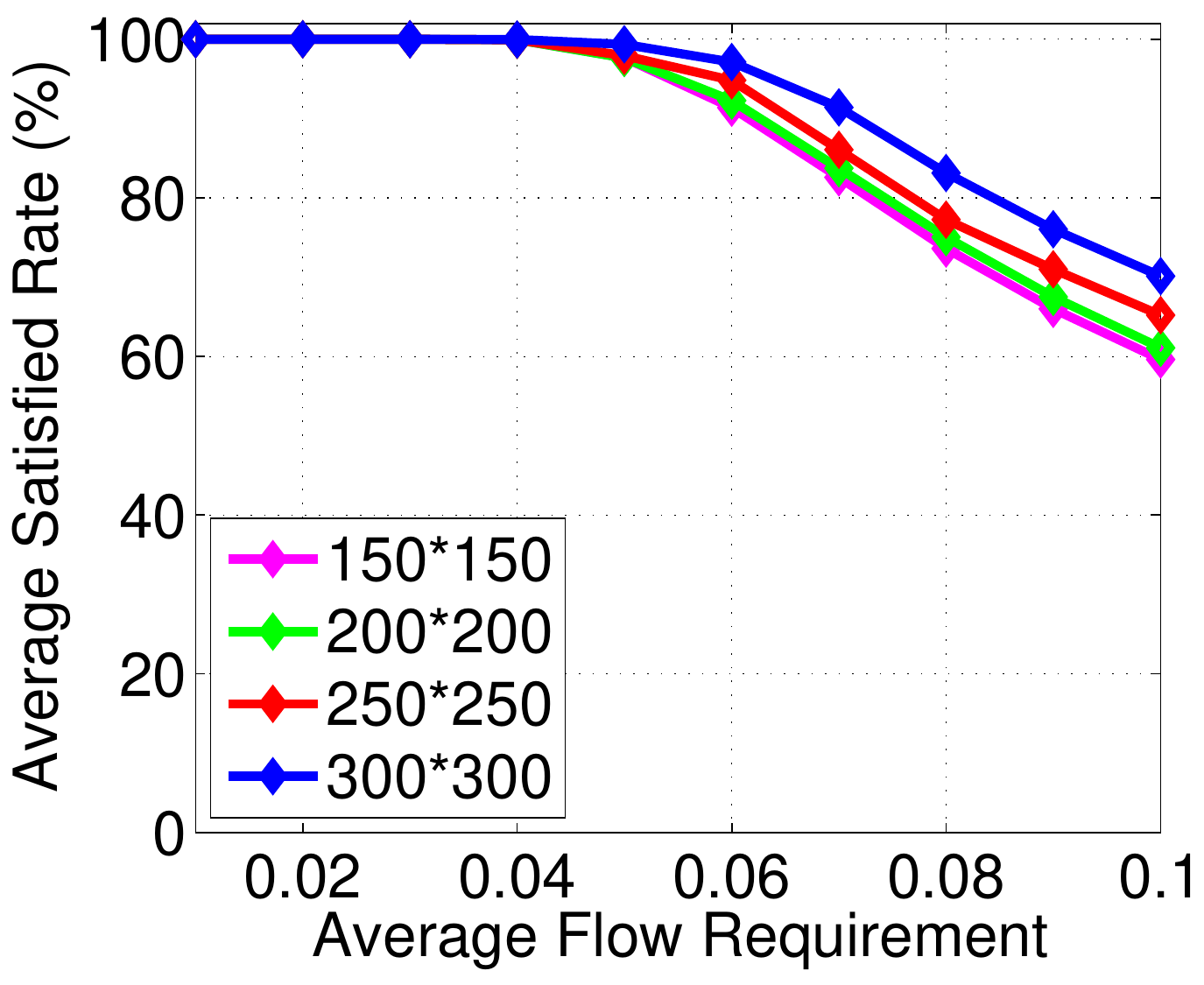}
\caption{\textrm{Scenario 1: Area Size}}\label{ef1}
\end{minipage}
\begin{minipage}[]{0.33\textwidth}
\centering\includegraphics[height=1.41in,width=1.9in]{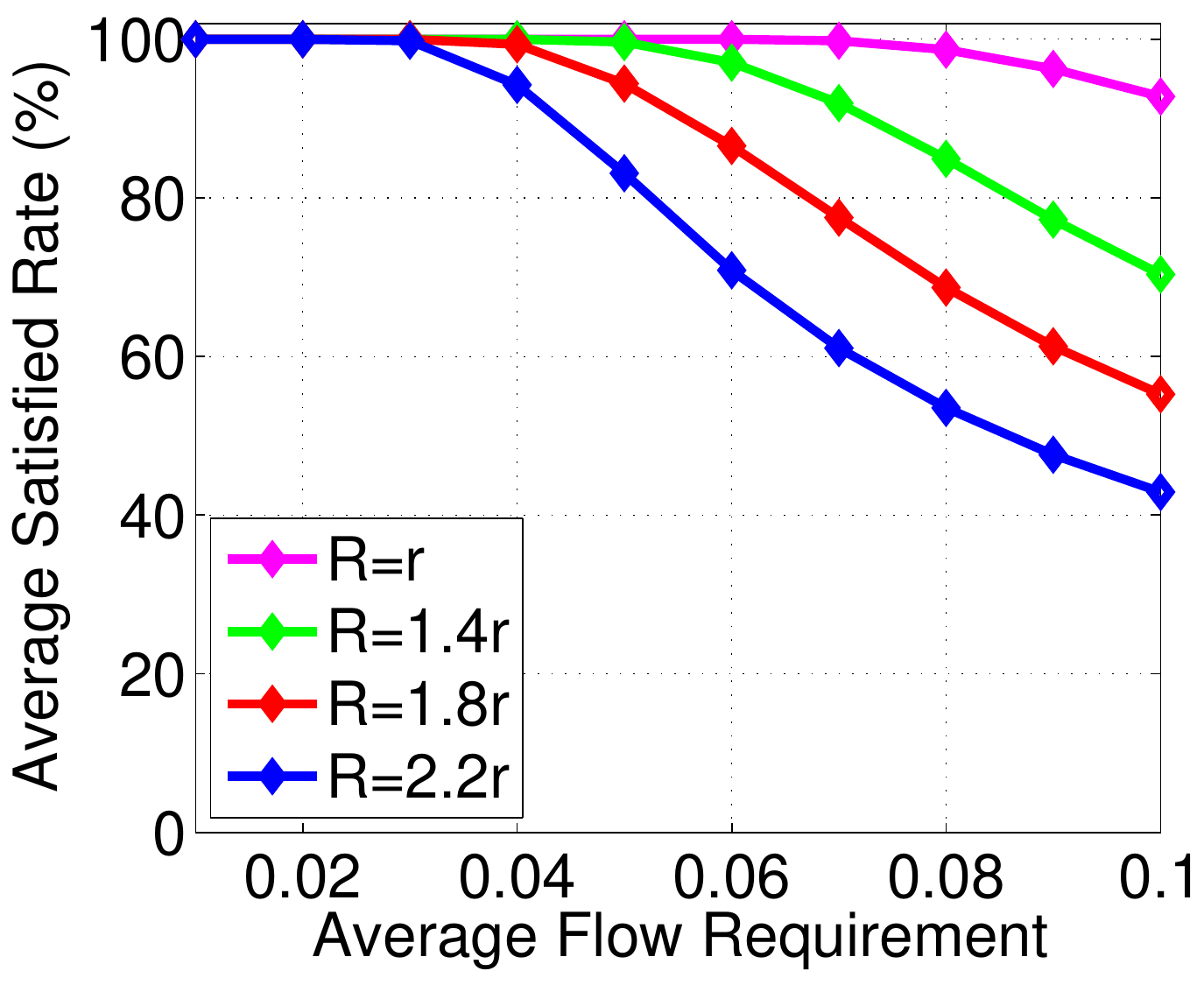}
\caption{\textrm{Scenario 1: Interference Range}}\label{ef2}
\end{minipage}
\begin{minipage}[]{0.33\textwidth}
\centering\includegraphics[height=1.41in,width=1.9in]{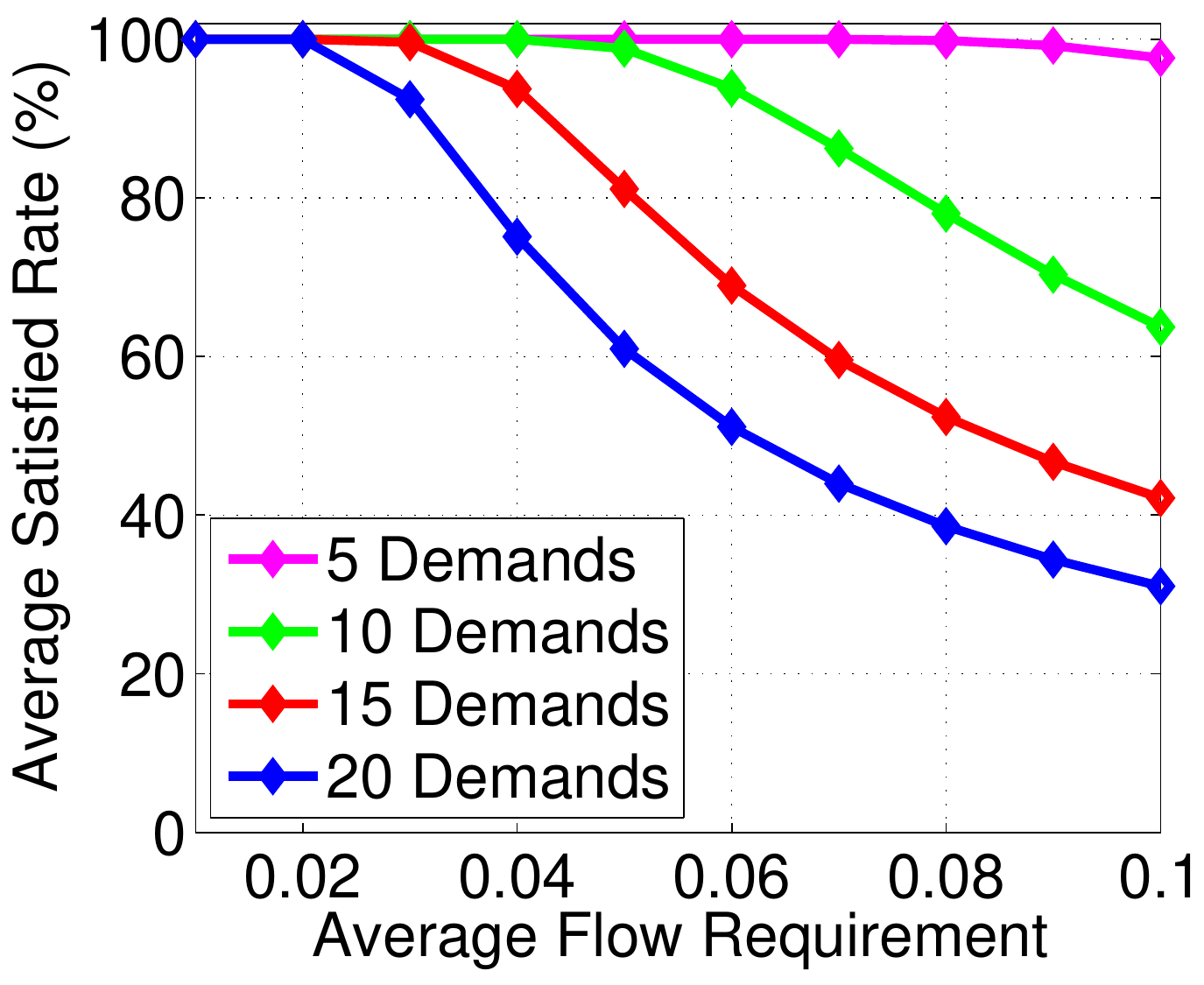}
\caption{\textrm{Scenario 1: Routing Demands}}\label{ef3}
\end{minipage}
\end{figure*}

\begin{figure*}[!htb]
\begin{minipage}[]{0.33\textwidth}
\centering\includegraphics[height=1.41in,width=1.9in]{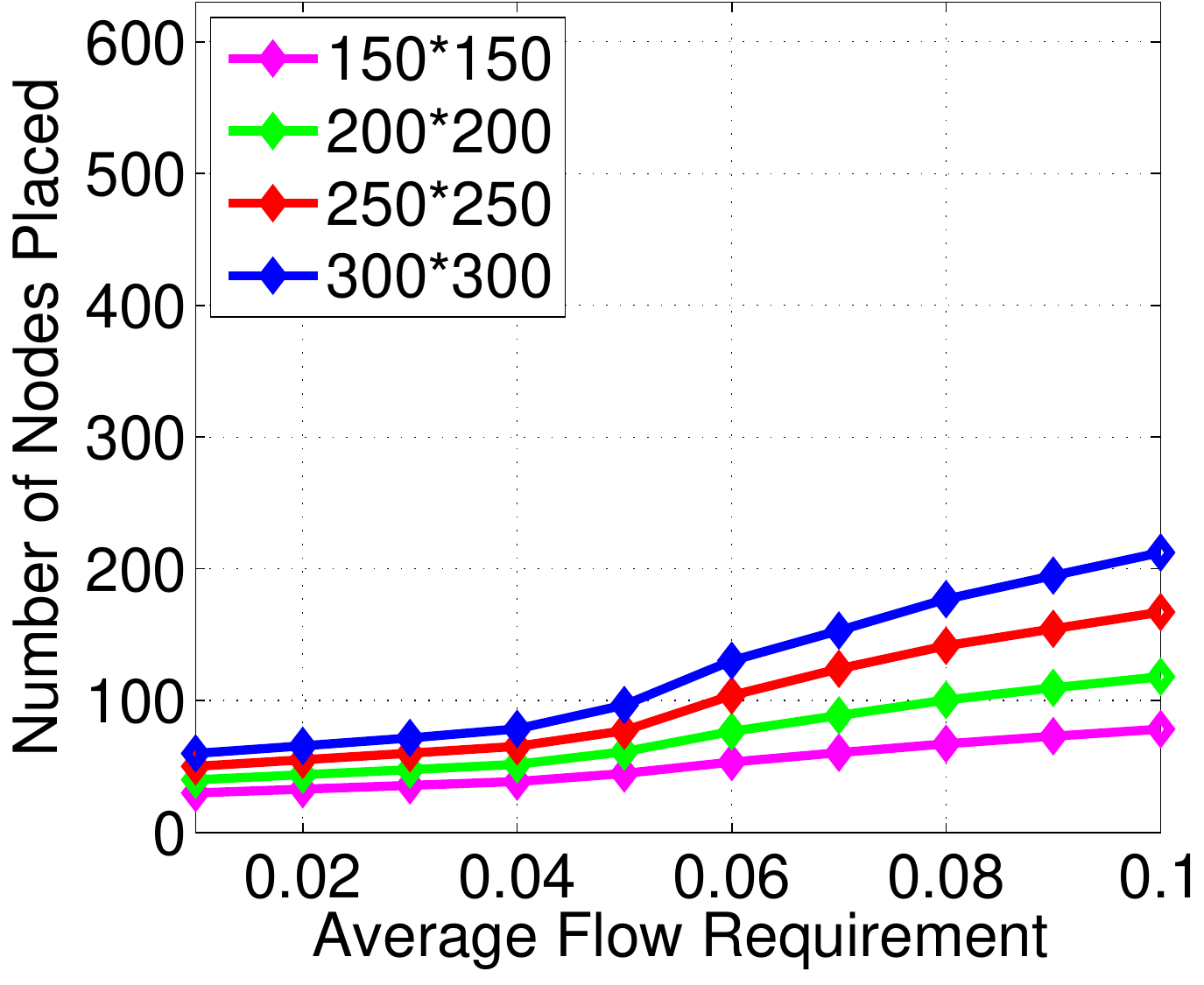}
\caption{\textrm{Scenario 1: Area Size}}\label{ef4}
\end{minipage}
\begin{minipage}[]{0.33\textwidth}
\centering\includegraphics[height=1.41in,width=1.9in]{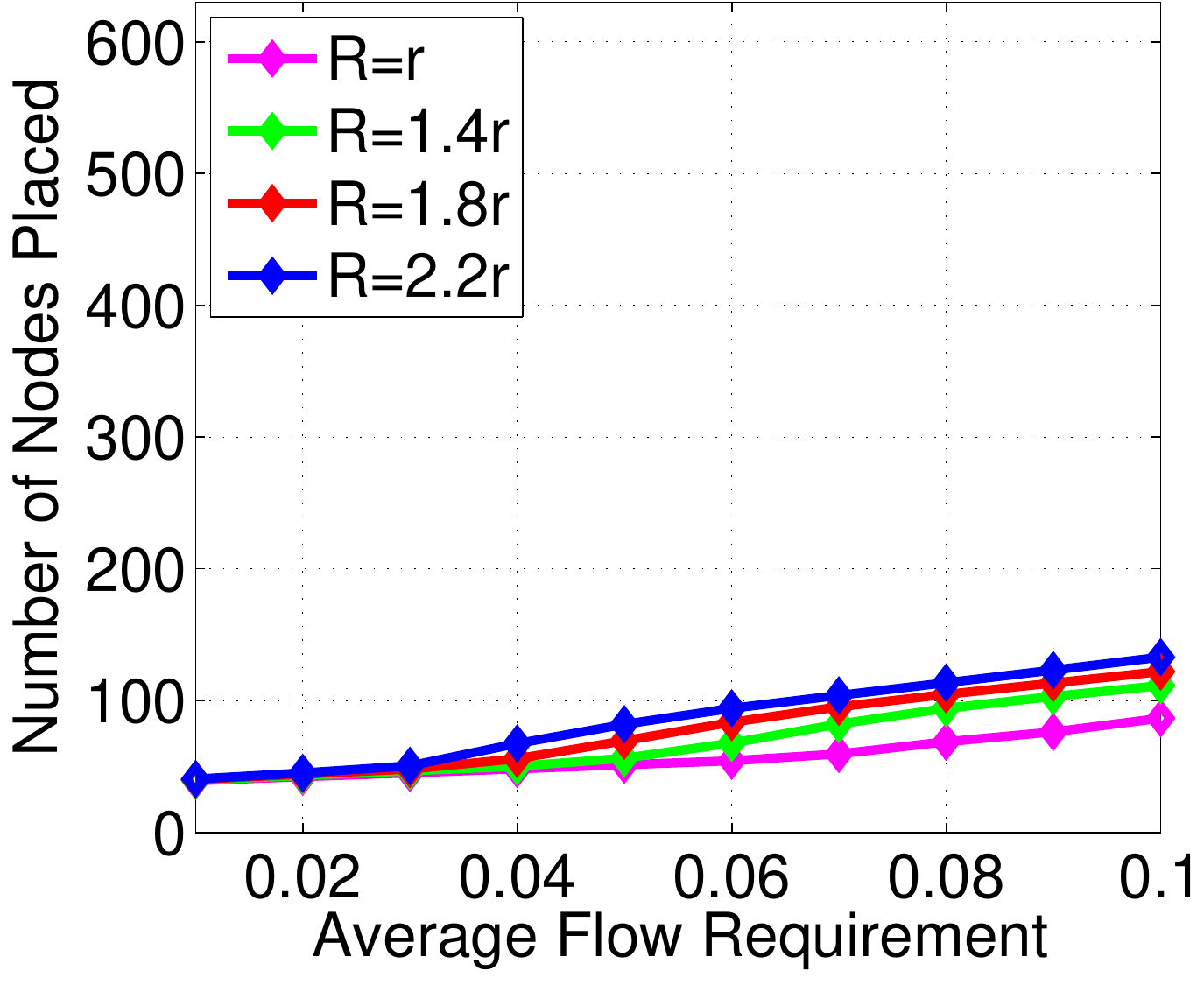}
\caption{\textrm{Scenario 1: Interference Range}}\label{ef5}
\end{minipage}
\begin{minipage}[]{0.33\textwidth}
\centering\includegraphics[height=1.41in,width=1.9in]{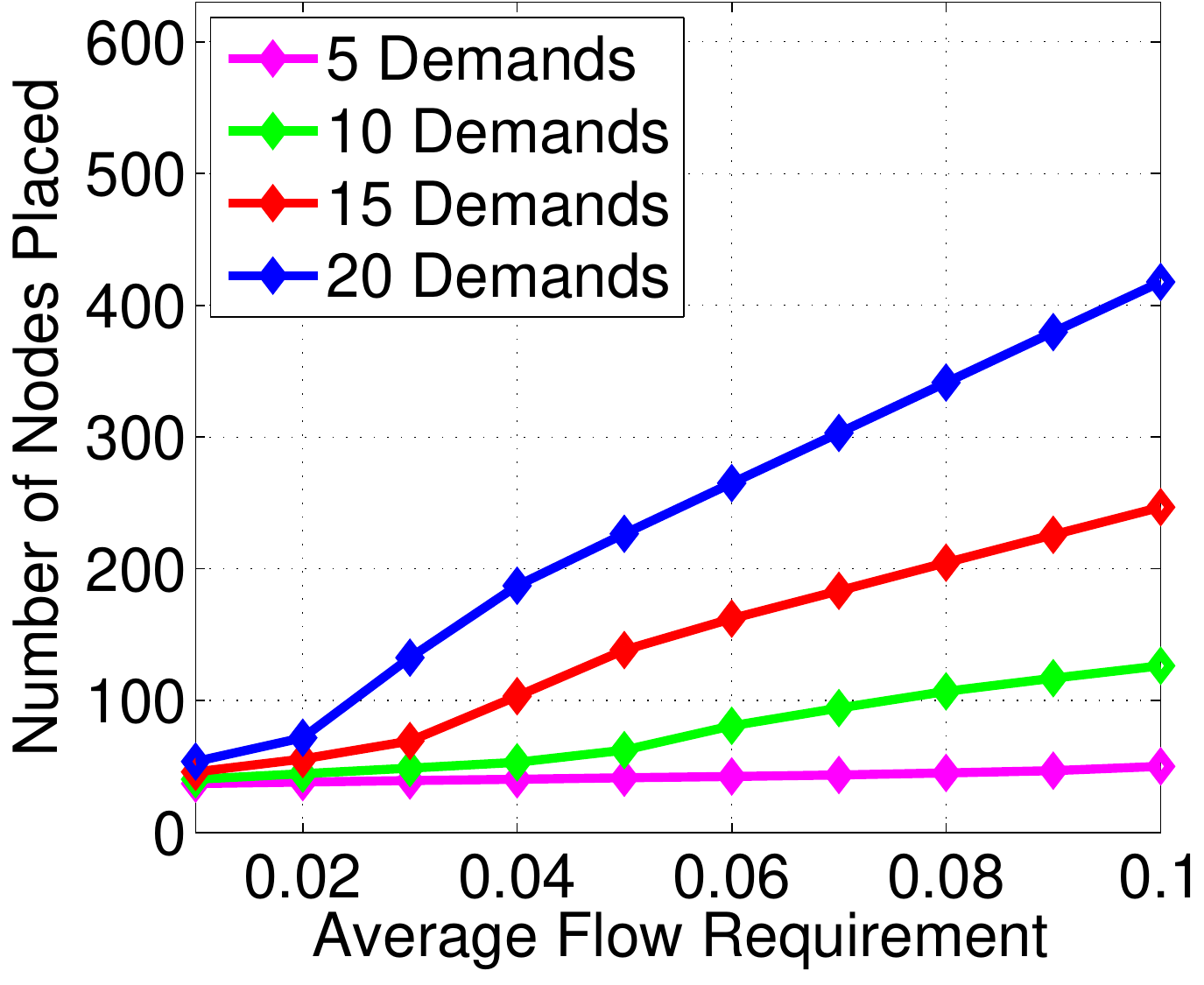}
\caption{\textrm{Scenario 1: Routing Demands}}\label{ef6}
\end{minipage}
\end{figure*}

\begin{figure*}[!htb]
\begin{minipage}[]{0.33\textwidth}
\centering\includegraphics[height=1.41in,width=1.9in]{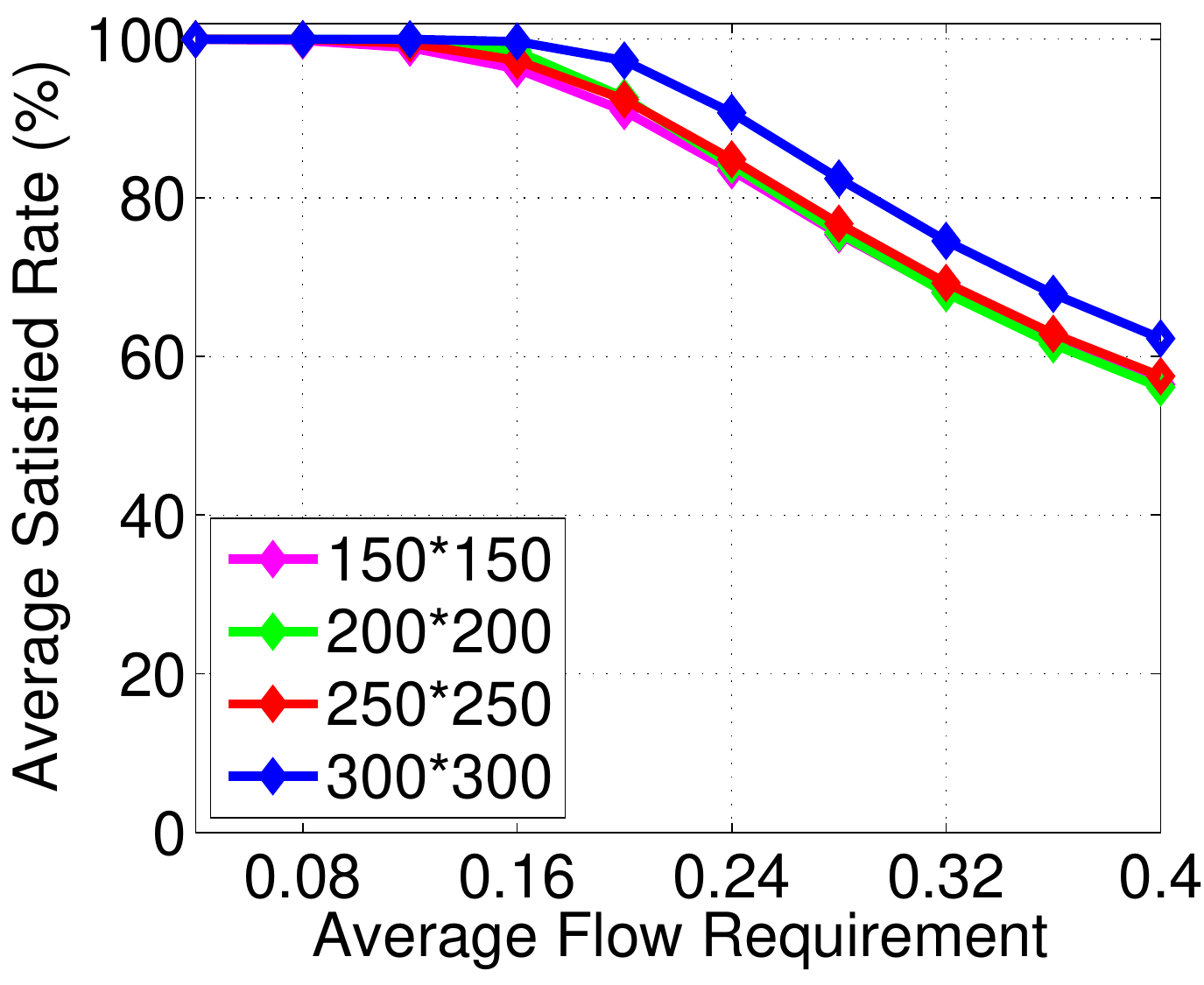}
\caption{\textrm{Scenario 2/3: Area Size}}\label{ef7}
\end{minipage}
\begin{minipage}[]{0.33\textwidth}
\centering\includegraphics[height=1.41in,width=1.9in]{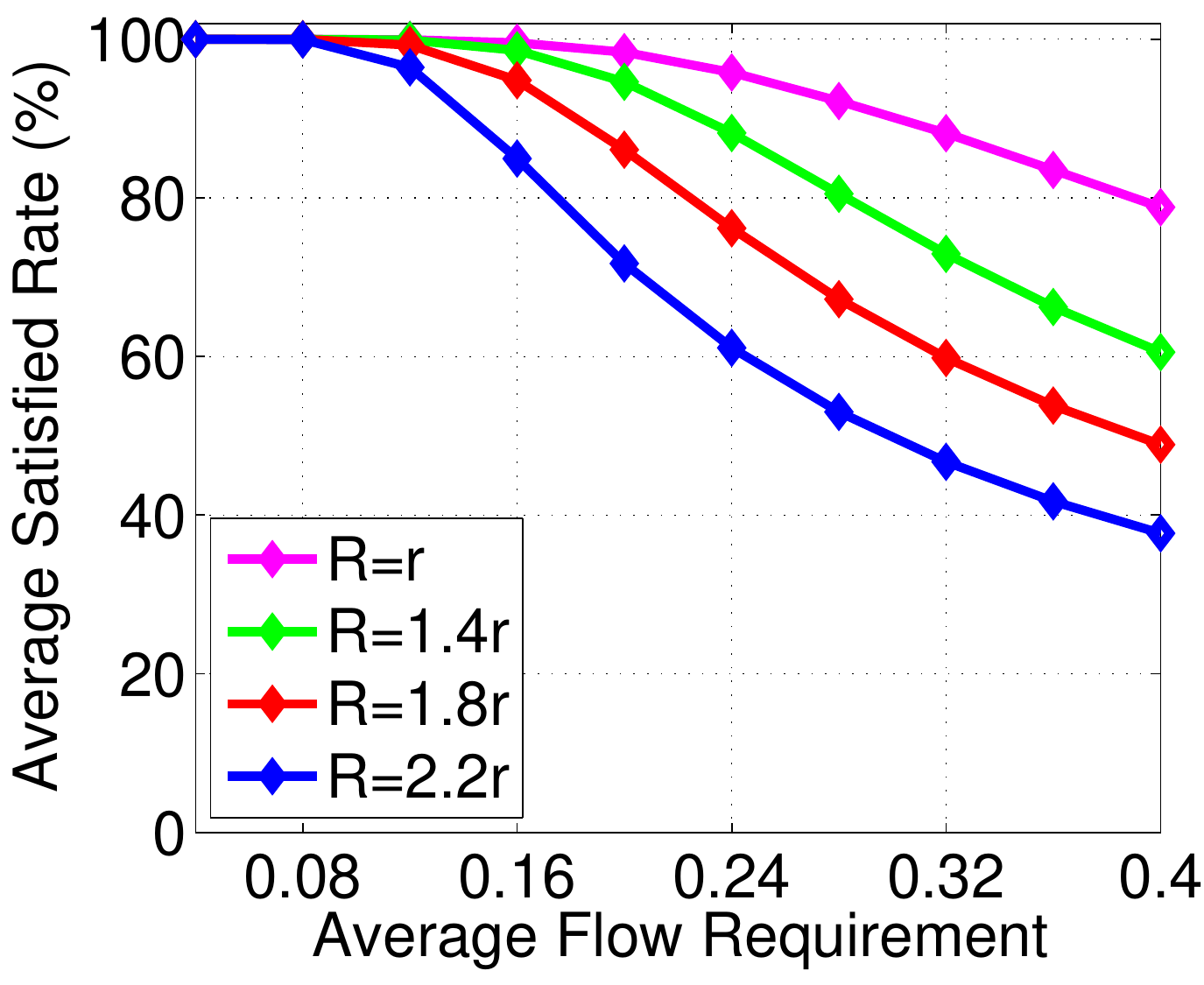}
\caption{\textrm{Scenario 2/3: Interference Range}}\label{ef8}
\end{minipage}
\begin{minipage}[]{0.33\textwidth}
\centering\includegraphics[height=1.41in,width=1.9in]{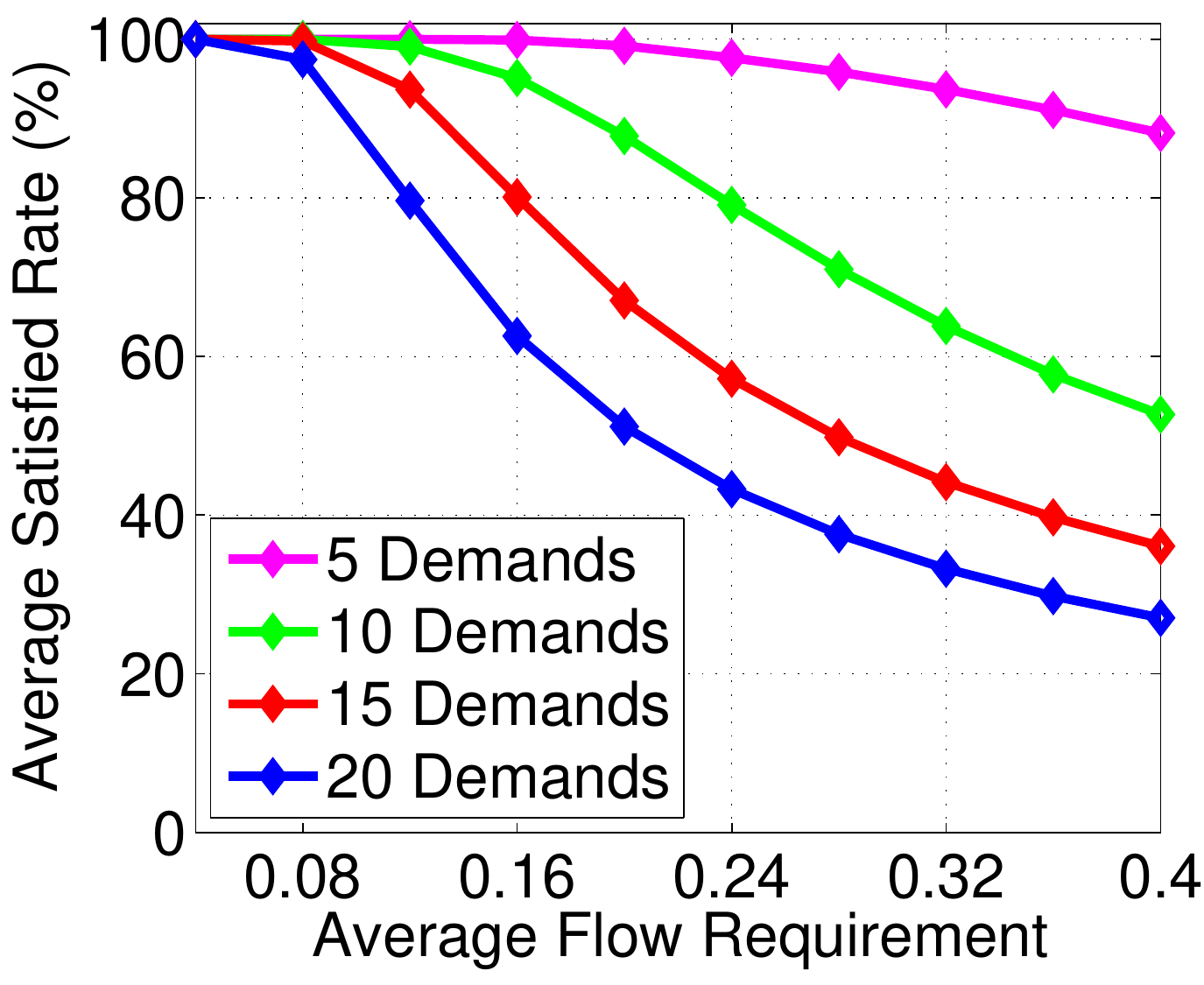}
\caption{\textrm{Scenario 2/3: Routing Demands}}\label{ef9}
\end{minipage}
\end{figure*}

\begin{figure*}[!htb]
\begin{minipage}[]{0.33\textwidth}
\centering\includegraphics[height=1.41in,width=1.9in]{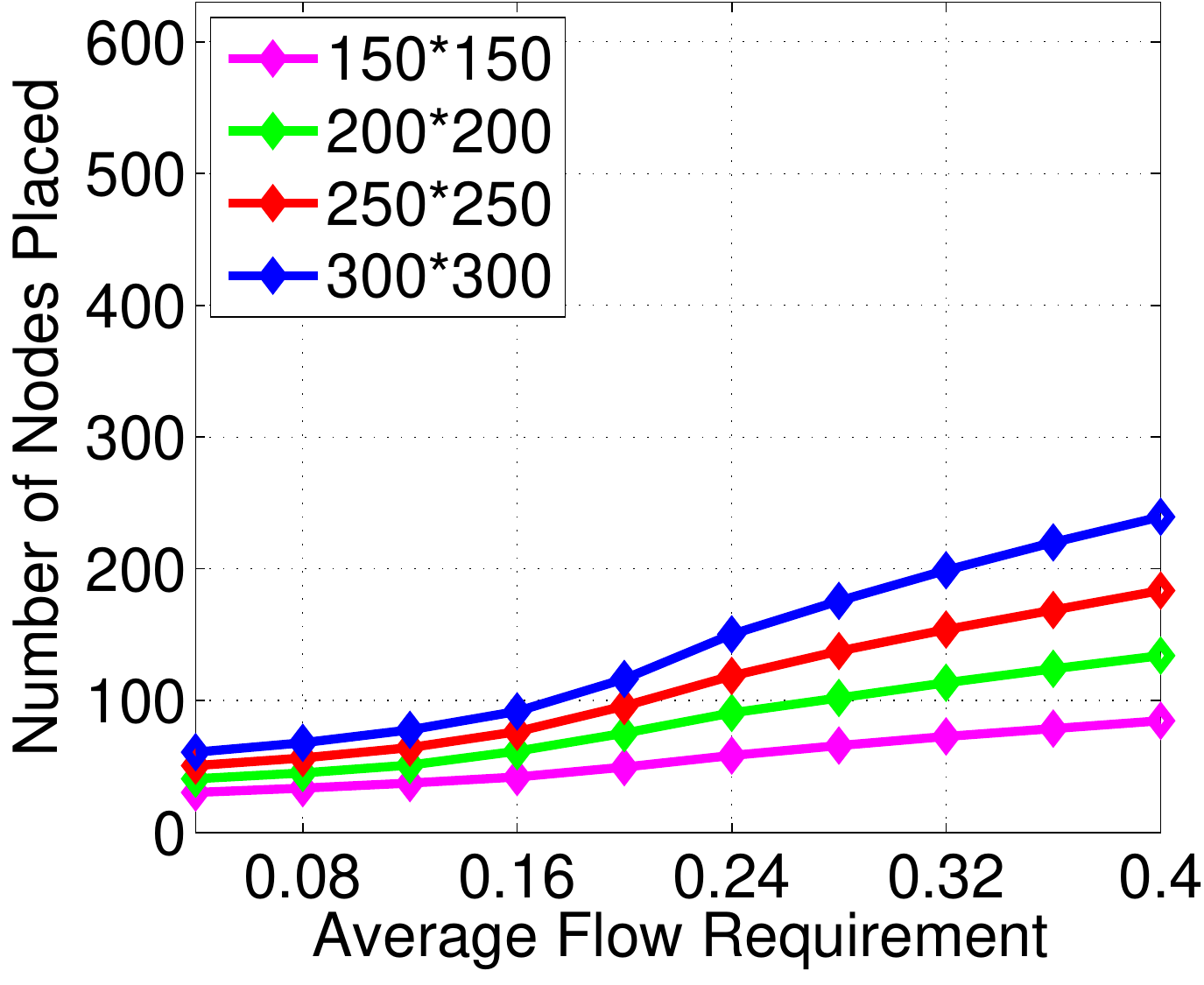}
\caption{\textrm{Scenario 2: Area Size}}\label{ef10}
\end{minipage}
\begin{minipage}[]{0.33\textwidth}
\centering\includegraphics[height=1.41in,width=1.9in]{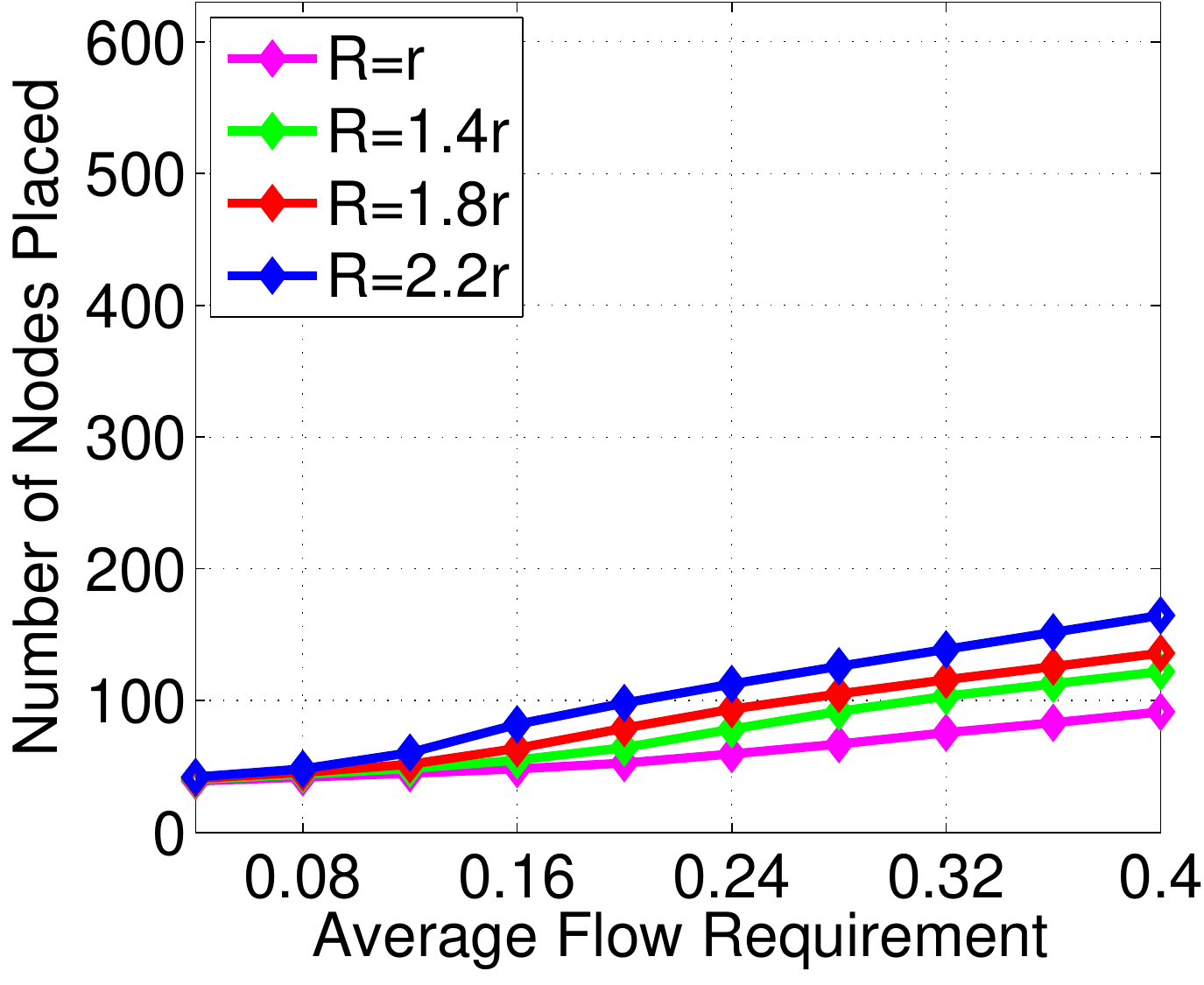}
\caption{\textrm{Scenario 2: Interference Range}}\label{ef11}
\end{minipage}
\begin{minipage}[]{0.33\textwidth}
\centering\includegraphics[height=1.41in,width=1.9in]{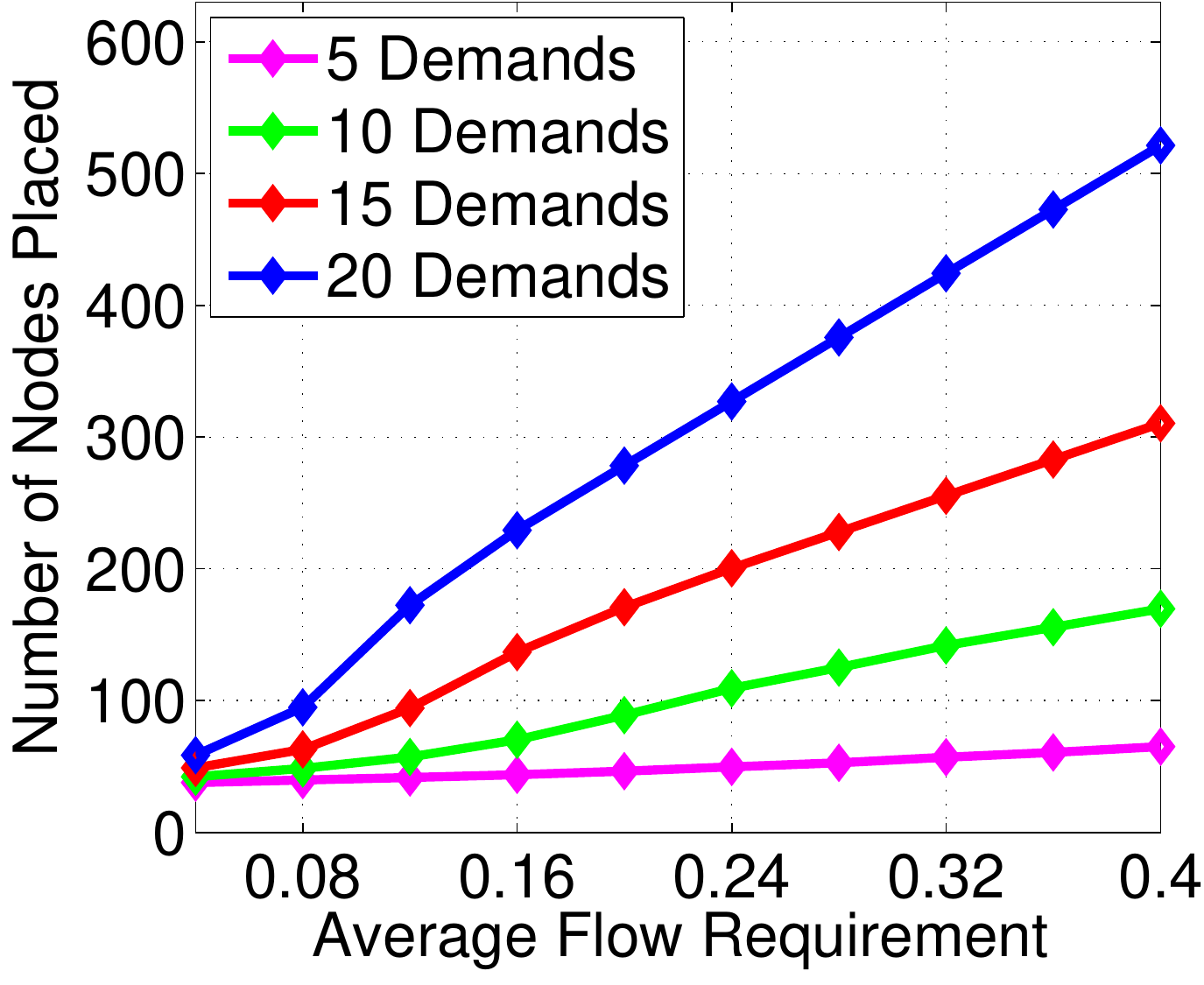}
\caption{\textrm{Scenario 2: Routing Demands}}\label{ef12}
\end{minipage}
\end{figure*}

\begin{figure*}[!htb]
\begin{minipage}[]{0.33\textwidth}
\centering\includegraphics[height=1.41in,width=1.9in]{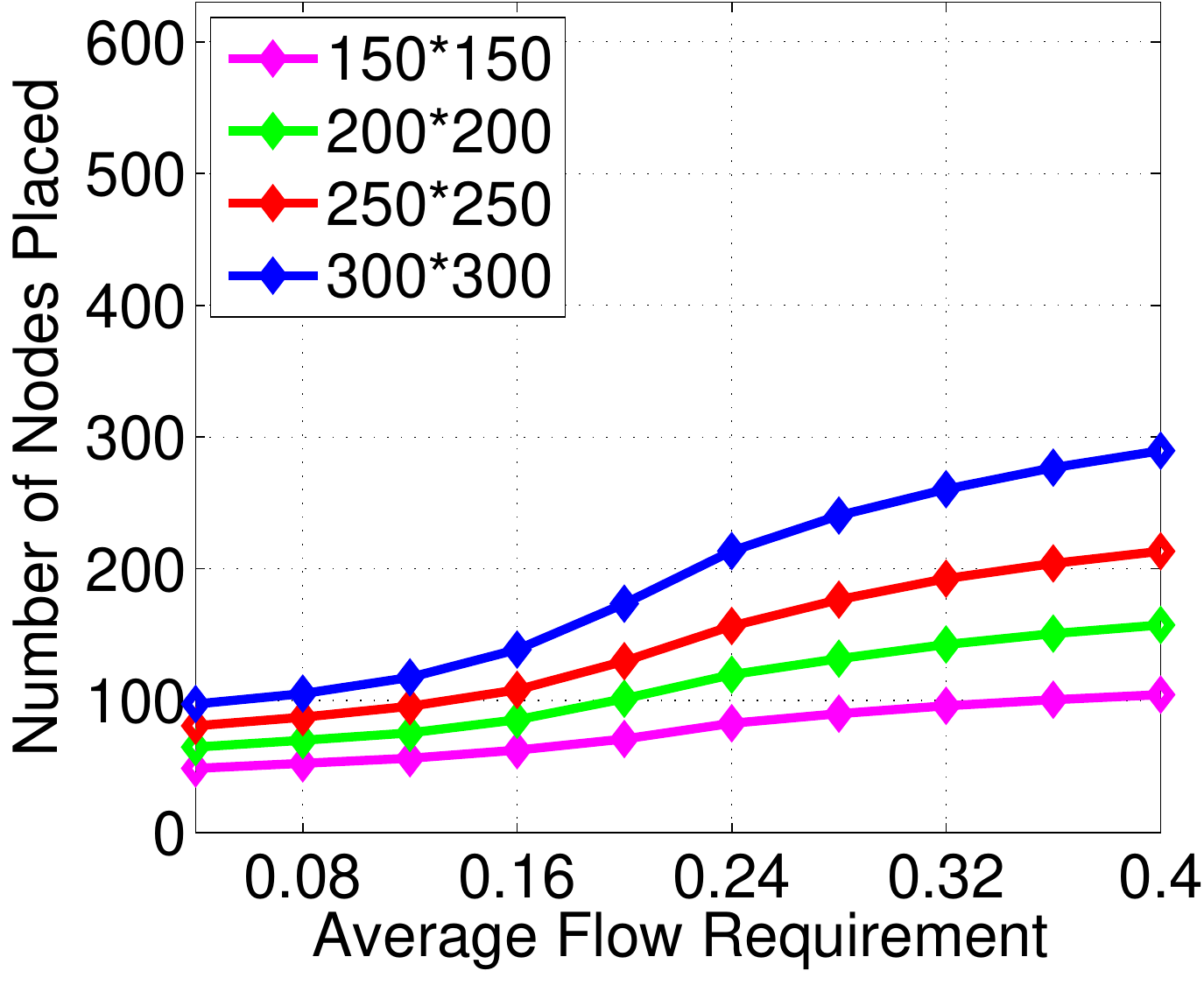}
\caption{\textrm{Scenario 3: Area Size}}\label{ef13}
\end{minipage}
\begin{minipage}[]{0.33\textwidth}
\centering\includegraphics[height=1.41in,width=1.9in]{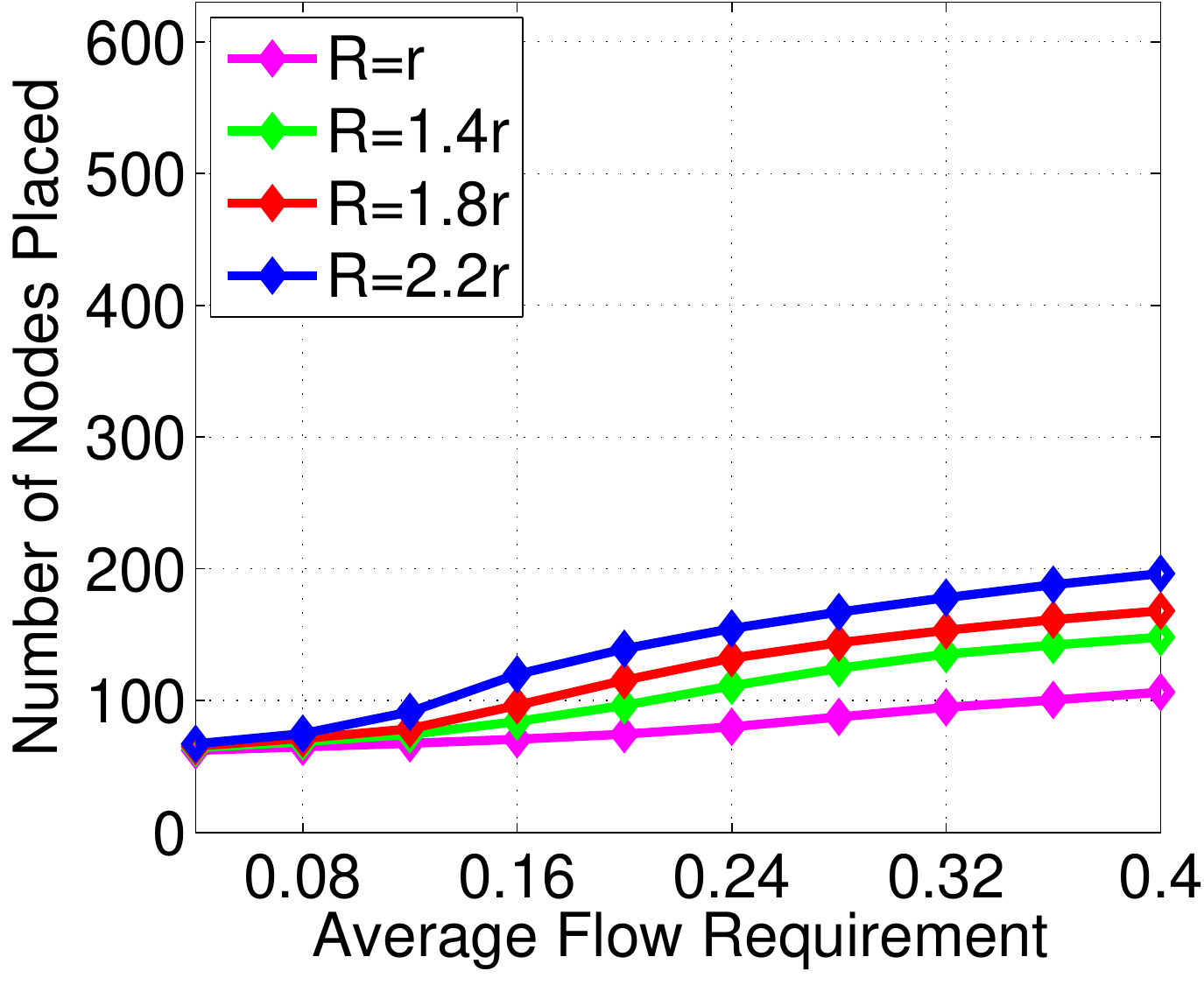}
\caption{\textrm{Scenario 3: Interference Range}}\label{ef14}
\end{minipage}
\begin{minipage}[]{0.33\textwidth}
\centering\includegraphics[height=1.41in,width=1.9in]{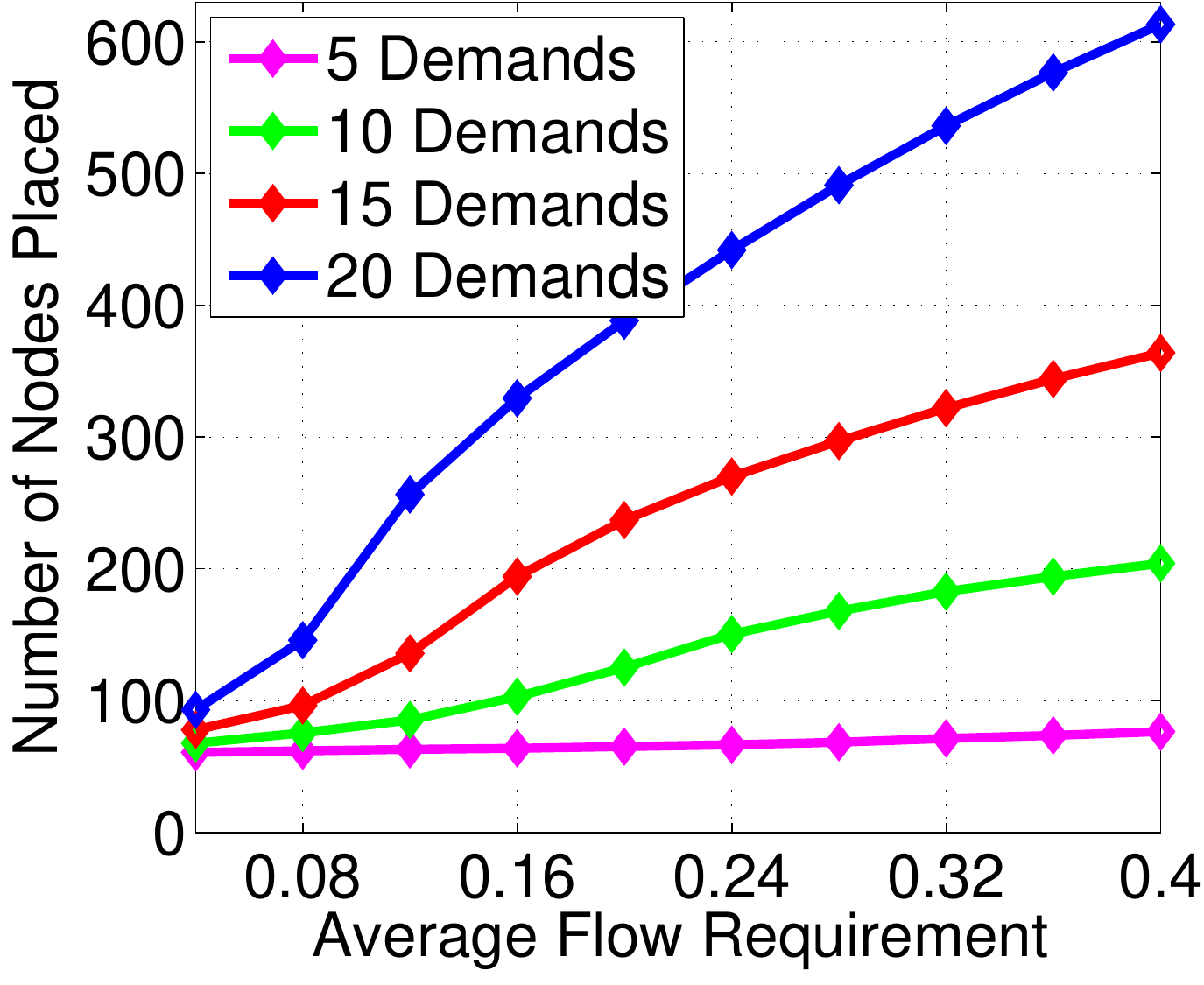}
\caption{\textrm{Scenario 3: Routing Demands}}\label{ef15}
\end{minipage}
\end{figure*}

\subsection{The Metric to Measure the Efficiency of Algorithms}
We evaluate algorithms proposed in Section \ref{node_placement} through simulations. Three scenarios, Data Aggregation, Demands with Definite Flow Requirement and Nodes with Unknown Flow Requirement, are used for evaluation. By default setting, nodes are placed in a $200 \times 200$ square region. There are $10$ routing demands. The transmission range $r$ is set to $10$ and the interference range $R$ is set to $\sqrt{2}r$. The flow $f$ is set to $1$. For each scenarios, the graph is randomly generated for $100$ times. In evaluation of these scenarios, the area size is changed in $150 \times 150$, $200 \times 200$, $250 \times 250$ and $300 \times 300$. The Interference range is changed $R=r$, $R=1.4r$, $R=1.8r$ and $R=2.2r$. The number of routing demands is changed in $5$, $10$, $15$ and $20$. We define the Satisfied Rate as
\begin{equation} \label{e6}
SR_q=\begin{cases}
\displaystyle 100\% & f^{'}_{src_qdest_q} \geq f_{src_qdest_q}  \\
\displaystyle \frac{f^{'}_{src_qdest_q}}{f_{src_qdest_q}} \times 100\% & f^{'}_{src_qdest_q} < f_{src_qdest_q}   \\
\end{cases}
\end{equation}
where $f_{src_qdest_q}$ is the flow requirement of $(src_q, dest_q)$ and $f^{'}_{src_qdest_q}$ is the flow actually can be achieved.

Based on equation \ref{e6}, we define the Average Satisfied Rate for $m$ routing demands as
\begin{equation} \label{e7}
ASR=\frac{\sum_{q=1}^{m} SR_q}{m}
\end{equation}
The higher the average satisfied rate can be achieved, the better the flow demands can be met.

\subsection{Experiments}

\subsubsection{Scenario 1: Data Aggregation}
In scenario 1, data are gathered from a set of source nodes and sent to a destination node. For this scenario, $m$ sources nodes and $1$ destination node are randomly generated in graph. The flow requirement is randomly generated for each source.

With the setting of flow $f=1$, the total flow of all source-destination pairs reaching the destination has no more than flow of $1$ for each time slot. Thus, the average flow requirement of each pair is generated from $0.01$ to $0.1$ in evaluation. The impacts of changing area size, interference range and routing demands on average satisfied rate are shown in figure \ref{ef1}, \ref{ef2} and \ref{ef3} respectively. The impacts of changing area size, interference range and routing demands on the total number of nodes used for placement are shown in figure \ref{ef4}, \ref{ef5} and \ref{ef6} respectively. The results are summarized later in section \ref{experiment_summary}.

\subsubsection{Scenario 2: Demands with Definite Flow Requirement}
In scenario 2, there is a set of $m$ routing demands $(src_q, dest_q)$ with definite flow requirement $f_{src_qdest_q}$, $q=1,2,...,m$. For this scenario, $m$ routing demands are randomly generated in graph. The flow requirement is randomly generated for each demand.

The average flow requirement of each pair is generated from $0.04$ to $0.4$ in evaluation. The impacts of changing area size, interference range and routing demands on average satisfied rate are shown in figure \ref{ef7}, \ref{ef8} and \ref{ef9} respectively. The impacts of changing area size, interference range and routing demands on the total number of nodes used for placement are shown in figure \ref{ef10}, \ref{ef11} and \ref{ef12} respectively. The results are summarized later in section \ref{experiment_summary}.

\subsubsection{Scenario 3: Nodes with Unknown Flow Requirement}
In scenario 3, there are a set of nodes that have been deployed, but we do not know the definite routing demands and flow requirement. For this scenario, a set of $2m$ nodes are randomly generated. $m$ nodes out of these nodes are randomly selected as source nodes and the other $m$ nodes are selected as destination nodes.

As the set of nodes are with unknown flow requirement, the paths constructed between source-destination pairs try to fulfill the flow of $F_C$ (Theorem \ref{thm_Fc}) using MPM algorithm (Algorithm \ref{alg2}). In evaluation, the average flow requirement of each pair is generated from $0.04$ to $0.4$. Actually, the results on average satisfied rate is the same as the results in scenario 2 (figure \ref{ef7}, \ref{ef8} and \ref{ef9}). The difference between scenario 3 and scenario 2 is that the flow requirement is unknown in scenario 3. Thus, the Merge function in algorithm \ref{alg4} does not merge paths in scenario 3 and more nodes need to be placed than scenario 2. The impacts of changing area size, interference range and routing demands on the total number of nodes used for placement are shown in figure \ref{ef13}, \ref{ef14} and \ref{ef15} respectively.

\subsection{Analysis} \label{experiment_summary}
The evaluation results are summarized as follows:

\begin{enumerate}
  \item \emph{The result of average satisfied rate.} The average satisfied rate is slightly higher with a larger area size, since a larger area has lower probability of interference between nodes. The average satisfied rate is lower with a larger interference range, since a larger interference range has higher probability of interference between nodes. The average satisfied rate is lower with more routing demands, since more routing demands with more nodes cause higher probability of interference between nodes. As can be seen from figure \ref{ef1}, \ref{ef2}, \ref{ef3}, \ref{ef7}, \ref{ef8} and \ref{ef9}, the average satisfied rate of flow demands reduces slowly when increasing the level of flow requirements.
  \item \emph{The result of total number of nodes used.} Slightly less nodes needed to be placed in scenario 1 than scenario 2, since more paths can be merged in data aggregation scenario. Less nodes needed to be placed in scenario 2 than scenario 3, since more nodes are needed to satisfy possible larger flow in scenario 3. (Average 25.4\%, 26.6\% and 24.1\% nodes haven been merged in figure \ref{ef10}, \ref{ef11} and \ref{ef12} than figure \ref{ef13}, \ref{ef14} and \ref{ef15} respectively.) The result verify the efficiency of merging routing paths.
\end{enumerate}

\section{Conclusion and Discussion}\label{conclusion}
The flow demands oriented node placement problem has been addressed in this paper. The problem is solved in three steps of calculating the maximal flow for single routing demand, calculating the maximal flow for multiple routing demands and finding the minimal number of nodes for routing demands with flow requirement. For single routing demand, we conduct its theoretical maximal flow can be achieved. For multiple routing demands, we prove both the problem of calculating the maximal flow and finding the minimal number of nodes for placement are NP-hard and propose polynomial-time complexity algorithms. The proposed algorithms are extensively evaluated in the scenarios of data aggregation, demands with definite flow requirement and nodes with unknown flow requirement, which verify its efficiency.

In future, we want to extend our work to a more general model setting. The synchronous time slotted system can be extended to asynchronous system by defining a probability model on transmission. The effect of delay and packet loss can also be introduced in the probability model. The studies on flow demands oriented node placement in this paper can still be applied with these changes.

\end{document}